%% file: final.tex
\documentclass[format=acmsmall, authorversion=true]{acmart}
\usepackage{acm-ec-20}
\usepackage{booktabs} 

\usepackage{amsmath}
\usepackage{amssymb}
\usepackage{eurosym}
\usepackage[shortlabels]{enumitem}

\AtBeginDocument{
\fancypagestyle{firstpagestyle}{
	\fancyhf{}%
	\fancyfoot[L]{\footnotesize{\emph{Unpublished draft manuscript. An extended abstract summarizing the results of the manuscript will appear at EC 2020.}}}%
}
\fancypagestyle{acmec}{
	\fancyhf{}%
	\fancyhead[C]{\footnotesize{\mbox{\shortauthors}}}
	\fancyhead[R]{\thepage}%
}
\pagestyle{acmec}
\renewcommand{\footnotetextcopyrightpermission}[1]{} 
\renewcommand{\footnotetextauthorsaddresses}[1]{} 
}

\newcommand{\playerone}{\textrm{P1}}
\newcommand{\playertwo}{\textrm{P2}}
\newcommand{\playeri}{\textrm{Pi}}
\newcommand{\playern}{\textrm{Pn}}
\newcommand{\playerj}{\textrm{Pj}}

\newcommand{\payoffkind}[1]{\ensuremath{\upsilon^{#1}}}
\newcommand{\lompayoffname}[0]{\ensuremath{\payoffkind{}}}
\newcommand{\lompayoffone}[2]{\ensuremath{\lompayoffname_1\left(#1,#2\right)}}
\newcommand{\lompayoffi}[2]{\ensuremath{\lompayoffname_i\left(#1,#2\right)}}

\newcommand{\insepsetA}[0]{\ensuremath{\mathcal{A}}}
\newcommand{\insepsetB}[0]{\ensuremath{\mathcal{B}}}

\newcommand{\halts}[1]{\ensuremath{#1 \!\! \downarrow}}

\DeclareMathOperator*{\argmax}{arg\,max}
\DeclareMathOperator*{\argmin}{arg\,min}

\setcitestyle{acmnumeric}

\title{A Complete Characterization of Infinitely Repeated Two-Player Games having Computable Strategies with no Computable Best Response under Limit-of-Means Payoff}

\author{Jakub Dargaj}
\affiliation{
\institution{University of Copenhagen}
\department{Computer Science}
}
\email{jada@di.ku.dk}
\author{Jakob Grue Simonsen}
\orcid{0000-0002-3488-9392}
\affiliation{
\institution{University of Copenhagen}
\department{Computer Science}
}
\email{simonsen@diku.dk}

\begin{abstract}
It is well-known that for infinitely repeated games, there are computable strategies that have
best responses, but no computable best responses. These results were originally proved for either specific
games (e.g., Prisoner's dilemma), or for classes of games satisfying certain conditions not known to be both
necessary and sufficient.

We derive a complete characterization in the form of simple necessary and sufficient conditions for the existence of a computable strategy without a computable best response under limit-of-means payoff. We further refine the characterization by requiring the strategy profiles to be Nash equilibria or subgame-perfect equilibria,
and we show how the characterizations entail that it is efficiently decidable whether an infinitely repeated
game has a computable strategy without a computable best response.
\end{abstract}

\begin{document}

\newtheorem{remark}[theorem]{Remark}

\maketitle

\section{Introduction}

We consider two-player games $G$ with simultaneous moves and perfect information. In a repeated game (or \emph{supergame}), $G$ is played repeatedly with all players aware of all moves played by all players in all previous games. The payoff of each player in such a game is a function of the payoffs obtained in the repetitions of $G$, for example the \emph{limit-of-means} payoff is the limit inferior of
the undiscounted averages of the payoff for each finite sequence of repetitions.
A \emph{computable} strategy for infinitely repeated games
is one where an algorithm computes the next action based on the finite history of previous repetitions of the game. 
Classic results from the 1990s show that infinitely repeated games admit computable strategies that have a best response, but no \emph{computable} best response \cite{bib:knoblauch1994,bib:Nachbar1996}, that is, some algorithm will play a strategy such
that there will exist a counterstrategy for the other player that will achieve maximum payoff among all strategies, but no such counterstrategy is computable. 
For infinitely repeated games with limit-of-means payoff, results are known solely for Prisoner's dilemma, and the
computable strategy involved is not known to be a Nash equilibrium \cite{bib:knoblauch1994}; for rational players, the absence of an equilibrium presents a problem: an algorithm might prevent other algorithms from obtaining maximal payoff, but possibly at the cost of not obtaining maximal payoff for itself. For infinitely repeated games with discounted payoff, results are known for a larger class of games containing Prisoner's dilemma that will ensure that the strategies involved form Nash or subgame-perfect equilibria \cite{bib:Nachbar1996}, but no necessary and sufficient conditions are known.

\emph{Computable} here means ``computable by a Turing machine''---the most general and widely accepted notion of what it means for a mathematical function to be computable \cite{Rogers,bib:sipser,Jones1997}. A Turing machine is an idealized notion of a computer that has a finite control (that is, a ``program''), but potentially limitless memory. Standard notions of restricted machines can typically be seen as Turing machines with restrictions on their running time or memory use (e.g., finite-state machines are Turing machines with constant memory). 
As a strategy in an infinitely
repeated game is map $s$ that, for any finite history (i.e., the finite sequence of previous actions played by both players in prior repetitions of $G$) outputs an action to be played in the next repetition, the strategy $s$ is computable if there exists a Turing machine that computes the map. Thus, a computable strategy $s$ that has a best response, but no \emph{computable} best response, is a strategy played by some (ordinary, finite) algorithm that when played against \emph{any} adversary that also plays according to some algorithm--that the adversary may choose freely--results in strictly suboptimal payoff for the adversary. However, as $s$ has a best response, an adversary with the ability to play a non-computable strategy--that is, a strategy that requires fundamentally more power to ``compute'' than what our current understanding of the term ``computer'' is able to--could, in principle, obtain optimal payoff.

\textbf{Contributions:}
For infinitely repeated games with limit of means payoff, we extend previous results in two directions: First, we identify  necessary and sufficient conditions for games to have computable strategies that have no computable best response, even though a best response exists; as a consequence of our techniques, we also provide necessary and sufficient conditions for strategies (computable or otherwise) to have no best response at all. Second, we obtain necessary and sufficient conditions for games to have such strategies in the case where the only strategies allowed are those that form Nash equilibrium, respectively a subgame-perfect equilibrium. 
In both cases, it is efficiently decidable whether a game satisfies the conditions

The general approach in our proof follows a standard technique in repeated games, namely using trigger strategies
that test for deviation from a prescribed path of play, entering a (finite or infinite) punishment phase ensured to decrease the opposing player's payoff and thus discouraging deviations from the prescribed path of play. All strategies use the notion of recursively inseparable sets, already utilized by Nachbar and Zame for discounted games
\cite{bib:Nachbar1996}. 

Both results \cite{bib:knoblauch1994,bib:Nachbar1996} make use of specific moves whose existence is guaranteed by the assumption that the game is (a variation of) the Prisoner's dilemma. The two key new insights are (i) that we can employ players' minmax payoff in $G$ in punishment phases to replace cooperation and defection from Prisoner's dilemma in almost all games, and (ii) that to establish an equilibrium, we can modify the strategies from the Folk theorems (standard results describing the set of equilibrium payoffs) by carefully incorporating the recursively inseparable sets in some repetitions of the game. Previously, strategies were either not required to form an equilibrium \cite{bib:knoblauch1994}, or used predefined moves (cooperation and defection) on the prescribed path \cite{bib:Nachbar1996}, while our strategies do not rely on these moves--instead, their existence is guaranteed by the Folk theorems.

The notion of minmax payoff plays an important role in the characterizations obtained. Consider a class of trivial games where a player cannot earn more than their minmax. Then it is rather apparent that a computable strategy played by any player has a computable best response in an infinitely repeated game.
Surprisingly, it turns out that any game, in order to have a computable strategy with no computable best response, just needs to allow a player to earn strictly more than their minmax; using different variations of strategies allow us to, essentially, use this criterion to also treat Nash and subgame-perfect equilibria.


\subsection{Related work}

Knoblauch proved, for limit-of-means payoff, that the Prisoner's dilemma
admitted computable strategies that have a best response, but no computable best response \cite{bib:knoblauch1994}, a result
later improved by Fortnow and Whang \cite{DBLP:conf/stoc/FortnowW94} showing that there is a polynomial-time computable strategy in Prisoner's dilemma that has no eventually $\epsilon$-optimal
computable response for any $\epsilon > 0$. 
Similarly, Nachbar and Zame show that for discounted payoff, there are computable strategies with best responses where no best response is computable for 
a class of two-player games that are paradoxical in the same way as the Prisoner's dilemma--rational players earn less than if they were both forced to make an irrational decision \cite{bib:Nachbar1996}. Unlike previous results for limit-of-means payoff, the strategies in \cite{bib:Nachbar1996} strategy are required to be subgame-perfect equilibria,
and the conditions for existence
of strategies without computable best responses are sufficient, but the authors conjecture that they are not \emph{necessary}. 

Both prior to, and after, the landmark results of Knoblauch and Nachbar and Zame, substantial work has been devoted to computing best responses (or Nash equilibria) for repeated games where strategies are constrained to be computable by machines with less power than the full Turing machines. Classic work includes Rubinstein \cite{RUBINSTEIN198683}, Gilboa \cite{GILBOA1988342}, Ben-Porath \cite{BENPORATH19901}, and Neyman and Okada \cite{DBLP:journals/ijgt/NeymanO00} (finite automata); Fortnow and Whang \cite{DBLP:conf/stoc/FortnowW94} (polynomial-time computable strategies). Modern results have mostly concerned variations
on the notion of equilibria or asymmetry between players, for example Chen et al.\ consider strategies with strictly bounded memories (a setting slightly different from strategies computable by finite automata) \cite{Chen-kmemory2017}, and Zuo and Tang
\cite{DBLP:conf/aaai/ZuoT15} study Stackelberg equilibria in a setting with restricted machines, and Chen et al.\ \cite{bib:Chen:2015} study changes to Nash equilibria of  infinitely repeated games under restrictions on the running time or space of the Turing machines. For games with discounted payoff, \cite{Berg2011} prove that all subgame-perfect equilibrium paths consist of elementary subpaths that can be represented as directed graphs.

Similar results concerning notions different from strategies that are known to exist classically, but fail to be computable exist elsewhere in Economics; for example, 
Richter and Wong show that there are exchange economies with all components computable and where a competitive equilibrium exists (by the Arrow-Debreu Theorem \cite{ArrowDebreu}), but no such equilibrium is computable \cite{RePEc:spr:joecth:v:14:y:1999:i:1:p:1-27}.


\section{Preliminaries}

We expect the reader to be familiar with basic notions from game theory and computability
theory at the level of introductory textbooks (e.g., \cite{bib:ft1991,Osborne1994,bib:Leyton-Brown:2008:EGT:1481632,bib:sipser}). To keep
the paper self-contained, we recap notation and some fundamental results in the following. Even though we are primarily interested in two-player games, we give definitions for games with any finite number of players in order to conform to standard notation.
We set $\mathbb{N} = \{1,2,\ldots\}$, $\mathbb{N}_0 = \{0\} \cup \mathbb{N}$, and we denote the set of rational numbers by
$\mathbb{Q}$ as usual.

\subsection{Game theory}

\begin{definition}[Normal-form game]
A (normal-form) game is a tuple \\ $(N, A, u)$ where:
\begin{enumerate}[(1)]
    \item $N = \{1,\ldots,n\}$ is the set of players (typically referred to as $\playerone, \ldots, \playern$).
    \item $A = A_1 \times \dots \times A_n$ is the set of \emph{action profiles}, where $A_i$ is a finite set of \emph{actions} available to \playeri{}.
    \item $u = (u_1, \dots, u_n)$, where $u_i: A \to \mathbb{Q}$ is the \emph{payoff} (aka.\ \emph{utility} or \emph{reward}) function for \playeri{}.
\end{enumerate}
\end{definition}

We shall mostly be interested in games $(N,A,u)$ with two players, that is, $N = \{\playerone{}, \playertwo{}\}$. For two-player games, the payoff function $u$ can be represented as a matrix with rows indexed by the actions available to \playerone{} ($A_1$), columns by $A_2$, and entries containing payoffs for each player when the corresponding action profile is played, separated by commas. This matrix is called a \emph{payoff matrix}.

We use the classic Prisoner's dilemma as a running example:
\begin{example}[Prisoner's dilemma]
\label{ex:prisonersdilemma}
Let $a, b, c, d \in \mathbb{R}$ satisfy $c > a > d > b$. Prisoner's dilemma is a~two-player game with $A_1 = A_2 = \{ C, D\}$ and the following payoff matrix:\\
\begin{center}
\begin{tabular}{ccc}
\textit{}                       & \textit{C}                         & \textit{D}                         \\ \cline{2-3} 
\multicolumn{1}{c|}{\textit{C}} & \multicolumn{1}{c|}{\textit{a, a}} & \multicolumn{1}{c|}{\textit{b, c}} \\ \cline{2-3} 
\multicolumn{1}{c|}{\textit{D}} & \multicolumn{1}{c|}{\textit{c, b}} & \multicolumn{1}{c|}{\textit{d, d}} \\ \cline{2-3} 
\end{tabular}
\end{center}
\end{example}
\begin{definition}[Pareto domination]
Let $G = (N, A, u)$ be a normal-form game.
Action profile $a$ is said to \emph{Pareto dominate} action profile $a'$ if
(i) For all $i \in N$, $u_i(a) \geq u_i(a')$, and (ii) there is $i \in N$ such that $u_i(a) > u_i(a')$.
If, for all $i \in N$, we have $u_i(a) > u_i(a')$, we say that $a$ \emph{strictly Pareto dominates} $a'$. 
\end{definition}
For example, in Prisoner's dilemma the action profile $(C,C)$ strictly Pareto dominates $(D,D)$. This follows from the initial assumption that for both players $i$, $u_i(C,C) = a > d = u_i(D,D)$. 

\begin{definition}
For an action profile $a = (a_1,\ldots,a_n)$ and \playeri, we denote by $a_{-i}$ the tuple
of actions of all other players, that is $(a_1,\ldots,a_{i-1},a_{i+1},\ldots,a_n)$.
\end{definition}

\begin{definition}[Best response; Nash equilibrium]
Let $G = (N,A,u)$ be a game, let $a = (a_1,\ldots,a_n) \in A$ be an action profile,
and let $a_i^* \in A_i$ be an action of \playeri{}.
We say that $a_i^*$ is a \textit{best response} to $a_{-i}$ if $u_i(a_i^*, a_{-i}) \geq u_i(a_i', a_{-i})$ for any other action $a_i' \in A_i$. We say that $a$ is a \textit{Nash equilibrium} of $G$ if, for all $i \in N$,
$a_i$ is a best response to $a_{-i}$.
\end{definition}






\subsection{Repeated games}

We now consider a situation when the same game is played infinitely many times; standard treatments of such games can be found in \cite{bib:aumann1981,bib:ft1991}, and we recapitulate basic terminology here.

\begin{definition}[Infinitely repeated game]
Given a game $G = (N, A, u)$, $G^\infty$ is a game which consists of infinitely many repetitions of the game $G$. $G$ is called the \emph{stage game} of the infinitely repeated game $G^\infty$.
\end{definition}
 Next we define a \emph{finite history} of length $T \in \mathbb{N}$ as a sequence of the first $T$ action profiles played in~$G^\infty$ and a \emph{path of play} as infinite sequence of action profiles.

\begin{definition}[Finite history]
Let $G^\infty$ be the infinitely repeated game of the stage game $G = (N, A, u)$. For a $T \in \mathbb{N}_0$, we shall write $\mathcal{H}_{G^\infty}^T = A \times \dots \times A = A^T$, and~$\mathcal{H}_{G^\infty} = \bigcup_{T \geq 0}{\mathcal{H}_{G^\infty}^T}$. A \emph{finite history of length} $T$ is any $h^T \in \mathcal{H}_{G^\infty}^T$.
\end{definition}


\begin{definition}[Path of play]
Let $G^\infty$ be the infinitely repeated game of the stage game $G = (N, A, u)$. We write $\mathcal{H}_{G^\infty}^\infty = A \times A \times \cdots = A^\infty$. A \emph{path of play} is any $h^\infty \in \mathcal{H}_{G^\infty}^\infty$.

For a finite history $h^T \in \mathcal{H}_{G^\infty}^T$ and $t \in \mathbb{N}$, $t \leq T$, we write $h^T_i[t]$ to denote the action played by Player $i$ in repetition $t$. Similarly, for a path of play $h^\infty$, $h^\infty_i[t]$ is the action played by \playeri{} in repetition $t$. We denote by $\bar{u}_i[t]=u_i(h^\infty[t])$ the payoff of \playeri{} in repetition $t$. 
\end{definition}


For example, let the stage game $G$ be Prisoner's dilemma from Example \ref{ex:prisonersdilemma} and consider the infinitely repeated game $G^\infty$. Since in every stage there are four action profiles available, there exist $4^T$ histories of length $T$. Assume that both players decide to play $C$ in odd stages and $D$ in even stages. This leads to the~path of play $h^\infty=\left((C,C), (D,D), (C,C), (D,D), \dots\right)$ and both players obtain the sequence of payoffs $(a,d,a,d,\dots)$. 

The payoff function for $G^\infty$ can be defined in multiple ways; in the present paper, we consider only
the \emph{limit-of-means} payoff (aka.\ \emph{average} payoff):
\begin{definition}[Limit-of-means payoff]
Given an infinite sequence of payoffs $(\bar{u}_i[1], \bar{u}_i[2], \dots)$ for \playeri{}, the \emph{limit-of-means payoff} of \playeri{} is defined as:
$$
\liminf_{T \to \infty}{\frac{1}{T}\sum_{t=1}^T{\bar{u}_i[t]}}.
$$
\end{definition}
Thus, if $G = (N,A,U)$, then any path of play of $G^\infty$ induces a limit-of-means payoff for each player. In games with limit-of-means payoff, the use of $\liminf$ ensures that any finite sequence of payoffs is ignored, so players seeking to maximize their payoff will only care about their behaviour in the infinite horizon. 

The action played by a player in the stage $t+1$ depends on the history of length $t$. All players have complete information about the actions played before, so a player's strategy maps finite histories into actions played in the next stage:
\begin{definition}[Strategy in a repeated game]
Let $G = (N,A,u)$ be a game.
A~\textit{(pure) strategy} for \playeri{} in $G^\infty$ is a map $s_i: \mathcal{H}_{G^\infty} \to A_i$.
A \emph{strategy profile} in $G^\infty$ is a tuple $s = (s_1, \dots, s_n)$
where, for each $i \in N$, $s_i$ is a strategy for \playeri{}.
\end{definition}
Observe that any strategy profile $s$ defines a unique path of play $h^\infty_s$, namely the one where
each player in stage $t \in \mathbb{N}$ of $G^\infty$ observes the finite history consisting of actions
played by all players in stages $1,\ldots,t-1$, and then use their strategy to play an action for stage $t$.
If $s = (s_1,\ldots,s_{i-1},s_i,s_{i+1},\ldots,s_n)$ is a strategy profile and $s'_i$ is a strategy for \playeri,
we write $(s'_i,s_{-i})$ for the strategy profile obtained by replacing $s_i$ by $s'_i$.

\begin{definition}[Payoff of a strategy profile]
Let $G = (N,A,u)$ be a game, let $s = (s_1,\ldots,s_n)$ be a strategy profile
in $G^\infty$, and let $h^\infty_s$ be the unique path of play induced by $s$. The (limit-of-means)
\emph{payoff} of \playeri{} is:
$$
\lompayoffname_i(s) = \liminf_{T \to \infty}{\frac{1}{T}\sum_{t=1}^{T}{u_i(h_s^\infty[t])}}
$$
\end{definition}

\begin{definition}[Best response; Nash equilibrium]
Let $G = (N,A,u)$ be a game, let $s = (s_1,\ldots,s_n)$ be a strategy profile in $G^\infty$,
and let $s_i^*$ be a strategy for \playeri{} in $G^\infty$.
We say that $s_i^*$ is a \textit{best response} to $s_{-i}$ if $\lompayoffi{s_i^*}{s_{-i}} \geq \upsilon_1(s_i', s_{-i})$ for any other strategy $s_i'$ for \playeri{}. We say that $s$ is a \textit{Nash equilibrium} if, for all $i \in N$,
$s_i$ is a best response to $s_{-i}$.
\end{definition}
For a two-player game and a strategy profile $s = (s_1,s_2)$ we abuse notation slightly by writing that
$s_1$ is a best response to $s_2$ instead of a best response to $s_{-1} = (s_2)$.  Observe that no player can \emph{unilaterally} choose an action (or strategy) that yields them a strictly better payoff than a Nash equilibrium--any strictly better payoff must involve other players changing strategies as well.

\begin{definition}[Subgame]
Let $G^\infty$ be an infinitely repeated game, $T \in \mathbb{N}$ and $h^T \in \mathcal{H}^T_{G^\infty}$. The subgame $(G^\infty, h^T)$ is the infinitely repeated game starting at~stage $T+1$ of $G^\infty$ with history $h^T$. 
\end{definition}
To illustrate the notion of a subgame, consider a~$G^\infty$ and a~strategy profile $s$ inducing the path of play $h^\infty_s$. If the history $h^T$ is a restriction of~$h^\infty_s$ to the first $T$ stages, then $s$ applied to the subgame $(G^\infty, h^T)$ leads to the path of~play $h^\infty_s[T+1,\dots]$, where $h^\infty_s[T+1,\dots]$ is the contiguous subsequence of $h^\infty_s$ starting at stage $T+1$. 
On the other hand, there may be histories containing actions that, according to $s$, are never played by any of the players. Every such history $h^{T'}$ defines a different subgame, and leads to a path of play that may have nothing in common with the original $h^\infty_s$.




\begin{definition}[Subgame-perfect equilibrium]
Let $G^\infty$ be an infinitely repeated game. A strategy profile $s$ is said to be a \emph{subgame-perfect equilibrium of} $G^\infty$ if it is a Nash equilibrium of every subgame.
\end{definition}



\subsection{Computability theory}

As usual, for any $A \subseteq \mathbb{N}$
we say that $A$ is \emph{recursively enumerable}
if there is a Turing machine that halts exactly on the elements of $A$ (equivalently, outputs exactly the elements of $A$),
and that $A$ is \emph{decidable} if there exists a Turing machine
 that halts on all inputs and accepts on input $n$ if{f} $n \in A$.

\begin{definition}
We assume a standard G{\"o}del numbering of the Turing machines
and denote by $T_m$  the $m$th Turing machine in this numbering, and by $\phi_m : \mathbb{N} \rightharpoonup \mathbb{N}$ the partial function computed by $T_m$. If $n \in \mathbb{N}$,
we write $\halts{\phi_m(n)}$ if $T_m$ halts
on input $n \in \mathbb{N}$. 
The \emph{jump} is the set
$\emptyset' = \{n \in \mathbb{N} : \halts{\phi_n(n)}\}$.
\end{definition}
The jump $\emptyset'$ is known to be recursively enumerable
and undecidable \cite[{\S}13.1]{Rogers}.
We shall use Smullyan's notion of recursive inseparability  \cite{bib:Smullyan1958}:

\begin{definition}
Let $\Sigma$ be a non-empty alphabet.
Sets $A, B \subseteq \Sigma^*$ are said to be \emph{recursively inseparable} if $A \cap B = \emptyset$ and there is no decidable set $C \subseteq \Sigma^*$ such that $A \subseteq C$ and $B \subseteq \Sigma^* \setminus C$.
\end{definition}
Observe that if $A$ is not decidable then $A$ and its complement are recursively inseparable. We use two standard sets
known to be recursively inseparable:
\begin{definition}
Define $\insepsetA = \{n \in \mathbb{N}: \halts{\phi_n(n)} \land \phi_n(n) = 0\}$, and 
$\insepsetB = \{ n \in \mathbb{N}: \halts{\phi_n(n)} \land T_n(n) \neq 0\}$.
\end{definition}
The following is well-known and provable by standard methods
(see, e.g.\ \cite{bib:Nachbar1996}):
\begin{proposition}
\label{prop:sets_AB}
Sets $\insepsetA$, $\insepsetB$ and $\insepsetA \cup \insepsetB$ are (i) recursively enumerable, (ii) undecidable, and (iii) recursively inseparable.
\end{proposition}

\begin{definition}
\label{def:AnBn}
For $n \in \mathbb{N}$, define:
\begin{align*}
\insepsetA_n &= \{ i \in \mathbb{N}: (i \leq n) \land  (T_i \mbox{ halts in at most } n-i \mbox{ steps on input } i)  \land (\phi_{i}(i) = 0)\} \\
\insepsetB_n &= \{ i \in \mathbb{N}: (i \leq n) \land (T_i \mbox{ halts in at most } n-i \mbox{ steps on input } i) \land (\phi_i(i) \neq 0) \} \qed
\end{align*}
\end{definition}
\begin{remark}\label{rem:setsABthenAnBn}
Observe that $\insepsetA_1 \subseteq \insepsetA_2 \subseteq \insepsetA_3 \subseteq \dots \subseteq \insepsetA$ and $\insepsetB_1 \subseteq \insepsetB_2 \subseteq \insepsetB_3 \subseteq \dots \subseteq \insepsetB$. Clearly, $\insepsetA_n$ and $\insepsetB_n$ are finite for all $n \in \mathbb{N}$ and hence decidable (even stronger: there exists a Turing machine that on input $n$ will
output (the G{\"o}del number of) a Turing machine deciding $\insepsetA_n$ because a universal Turing machine
can simulate at most $n$ steps of $T_i$ on input $i$; similarly for $\insepsetB_n$).
Observe also that for $n \in \insepsetA$, there is some $k \in \mathbb{N}$ such that
$T_n$ halts in $k$ steps on input $n$, whence $n \in \insepsetA_{n+k}$.
\end{remark}

\begin{definition}
A pure strategy $s_i: \mathcal{H}_{G^\infty} \to A_i$ for \playeri{} is \emph{computable}
if there is a Turing machine that, on input a finite history $h \in \mathcal{H}_{G^\infty}$ 
(represented by some element of $\{0,1\}^*$) halts with output $s_i(h)$ (represented by some element of $\{0,1\}^*$).
\end{definition}

%
%




\section{Non-trivial games and best responses}

Consider a 2-player normal-form game $G$ and its infinite repetition $G^\infty$ with limit-of-means payoff. 

\begin{definition}
Let $G$ be a 2-player normal-form game and $a_{-i}$ be an action available to Player $-i$. We define
$M_i(a_{-i}) = \max_{a_{i} \in A_{i}}{u_i(a_i, a_{-i})}$,
and
$M_i = \max_{a \in A}{u_i(a)}$. \qed
\end{definition}

Suppose that the payoff of the best response of \playerone{} is independent of the action played by \playertwo{}, that is $\forall a_{2} \in A_{2}: M_1 = M_1(a_{2})$. This is equivalent to saying that no action gives Player $1$ higher payoff than their minmax payoff, and we will call such games trivial for \playerone.

\begin{definition}
Let $G$ be a 2-player normal-form game. Then, $G$ is said to be \emph{trivial for} \playeri{} if
$M_i = \min\limits_{a_{-i} \in A_{-i}}\max\limits_{a_{i} \in A_{i}} u_i(a_i, a_{-i})$. $G$ is said to be
\emph{non-trivial for} \playeri{} if it is not trivial for \playeri{}.
\qed
\end{definition}

For example, Prisoner's Dilemma  is non-trivial for any player; an example of a game that is trivial for any player
is Rock-Paper-Scissors (see Example \ref{ex:rockpaperscissors}). 

If a game is trivial for a player, that player will always have a best response to any strategy; moreover, the best response to a strategy requires no more computational resources than the original strategy, as it needs only scan the correct row (or column) of the payoff matrix and play the action maximising their profit in the current stage:

\begin{lemma}
\label{lemma:avg_trivial_br}
Let $G$ be trivial for \playerone{}. Then, under limit-of-means payoff:
\begin{enumerate}
    \item Every strategy of \playertwo{} has a best response.
    \item Every computable strategy of \playertwo{} has a computable best response.
\end{enumerate}
\end{lemma}
\begin{proof}
Let $s_2$ be any strategy of \playertwo{}. Define $s_1$ to be the strategy of \playerone{} that, given a finite history $h^T \in \mathcal{H}^T_{G^\infty},$ in stage $T+1$ computes $a_2 = s_2(h^T)$ and plays $a_1 = \argmax\limits_{a'_1 \in A_1}{u_1(a'_1,a_2)}.$ Because $G$ is trivial for \playerone{}, $u_1(a_1, a_2) = M_1.$ \playerone{}'s limit-of-means payoff when playing $s_1$ is: 
$$
\lompayoffone{s_1}{s_2}
 = \liminf_{T \to \infty}{\frac{1}{T}\sum_{i=1}^T{u_1(h_1^\infty[i], h_2^\infty[i])}} = M_1.
$$
Because $M_1$ is the maximum payoff \playerone{} can obtain in $G$, it is also the maximum limit-of-means payoff \playerone{} can obtain in $G^\infty$, and hence $s_1$ is a best response to $s_2$. If $s_2$ is computable, then $s_1$ clearly computable as the set of available actions $A_1$ is finite.  
\end{proof}

Hence, non-triviality is a \emph{necessary} condition for the existence of strategies without a best response, and of computable strategies without a computable best response. It turns out that it is also a \emph{sufficient} condition. 

\begin{definition}
Let G be non-trivial for Player $1$. Define $C_1, D_1 \in A_1$ and $C_2, D_2 \in A_2$ to be any actions satisfying (1) $u_1(C_1, C_2) = M_1$, and (2) $u_1(D_1, D_2) = M_1(D_2) < M_1$.\qed
\end{definition}

The action profile $(C_1, C_2)$ gives \playerone{} the maximum possible payoff in $G$. The existence of $(D_1, D_2)$, where $D_1$ is a best response to $D_2$, but \playerone{} obtains a lower payoff than from $(C_1, C_2),$ is guaranteed by non-triviality of $G$. We intentionally use the same notation as for Prisoner's dilemma to differentiate between the high-payoff and low-payoff action profiles, so that the strategies defined in this section are reminiscent of the strategies from \cite{bib:knoblauch1994}. However, we do not--at the moment--require $(D_1, D_2)$ to be a Nash equilibrium of $G.$ 

\subsection{Every non-trivial game has a strategy having no best response}

We now define a computable strategy that does not admit a best response (computable or otherwise).

\begin{definition}\label{def:the_d_strategy}
Let $G$ be non-trivial for \playerone{}. Define $\sigma^d_{2}$ to be \playertwo{}'s strategy in $G^\infty$ that, given a finite history $h^T \in \mathcal{H}^T_{G^\infty},$ plays the following action in stage $T+1:$
\begin{enumerate}[(1)]
    \item Play $D_2$ if Player $1$ has never played $C_1$ in $h^T$.
    \item If Player $1$ has played $C_1$ in $h^T,$ let $t$ be the first stage when Player $1$ plays $C_1$. If $(t+1)$ divides $(T+1)$, play $D_2$, otherwise play $C_2$.
\end{enumerate}
\end{definition}

\begin{lemma}\label{lem:non_trivial_no_best_response}
$\sigma^d_2$ is a computable strategy.
If $G$ is non-trivial for \playerone{}, then $\sigma^d_2$ has no best response.
\end{lemma}
\begin{proof}
$\sigma^d_2$ is clearly computable: A Turing machine can scan the finite history $h^T$ to
find whether \playerone{} has played $C_1$ at any stage. If so, the first such stage $t$ can be found in finite time, and it is clearly decidable whether $t+1$ divides $T+1$.

Now, let $s_1$ be any strategy for \playerone{},
and let $h^\infty$ be the path of play induced by the strategy profile $s=(s_1, \sigma^d_{2})$, and let $h^\infty_i[T]$ be the action played by Player $i$ in stage $T$.
Split on cases as follows:
\begin{itemize}
    \item \playerone{} plays $C_1$ in at least one stage of $G^\infty$. Let $t$ be the first stage where \playerone{} does so. Then, \playerone{}'s payoff is:
     $$
    \lompayoffone{s_1}{\sigma^d_{2}} 
    = \liminf_{T \to \infty}{\frac{1}{T}\sum_{i=1}^T{u_1(h^\infty_1[i], h^\infty_2[i])}}  
    = \liminf_{T \to \infty}{\frac{1}{T}\sum_{i=t+1}^T{u_1(h^\infty_1[i], h^\infty_2[i])}} 
    \leq
    \frac{M_1(D_2) + t M_1}{t+1}    
    $$
 where the final inequality follows from the fact that \playerone{}'s maximum payoff in any stage where \playertwo{} plays
    $D_2$ is $M_1(D_2)$ which happens with frequency $1/(t+1)$ at each stage after $T$; similarly, \playertwo{} plays $C_1$
    with frequency $t/(t+1)$ after stage $T$ (every stage where $t+1$ does not divide $T+1$), and in
    each stage where \playertwo{} plays $C_1$, \playerone{}'s payoff is at most $M_1$.
    
    Let $s_1'$ be the strategy for \playerone{} that plays $C_1$ for the first time in stage $t+1$ (and plays any other action in the first $t$ stages); for $T > t$, in stage $T+1$, \playerone{} plays $D_1$ if $t+2$ divides $T+1$, and otherwise plays $C_1$. Then, by the same reasoning as above:
    $$ 
    \lompayoffone{s_1'}{\sigma^d_{2}} = 
    \frac{M_1(D_2) + (t+1) M_1}{t+2} > \lompayoffone{s_1}{\sigma^d_{2}}
    $$
    Thus, the strategy $s_1$ is not a best response to $\sigma^d_{2}$.
    
    \item \playerone{} does not play $C_1$ in any stage of $G^\infty$. Then, by the definition
    of $\sigma_2^d$, \playerone{}'s payoff is:
    $$ 
    \lompayoffone{s_1}{\sigma^d_2} \leq M_1(D_2)
    $$
    Consider the strategy $s_1'$ for \playerone{} that plays $C_1$ in odd-numbered stages and $D_1$ in even-numbered stages. The strategy profile $(s_1', \sigma^d_2)$  has path of play 
    $$
    h^\infty = ((C_1,D_2), (D_1,D_2), (C_1,C_2), (D_1,D_2), (C_1,C_2), \dots)
    $$ 
    and \playerone{}'s payoff is thus:
    \begin{align*}
    \lompayoffone{s_1'}{\sigma^d_2} 
    &= \liminf_{T \to \infty}{\frac{1}{T}\sum_{i=1}^T{u_1(h^\infty_1[i], h^\infty_2[i])}}  
    = \liminf_{T \to \infty}{\frac{1}{T}\sum_{i=2}^T{u_1(h^\infty_1[i], h^\infty_2[i])}} \\ 
    &= \frac{1}{2}(M_1(D_2) + M_1) 
    > M_1(D_2)
    \geq \lompayoffone{s_1}{\sigma^d_2}    
    \end{align*}
    and thus $s_1$ is not a best response to $\sigma^d_2$.
\end{itemize}
Thus, for every choice of strategy for \playerone{}, there exists another strategy obtaining
better payoff against $\sigma_2^d$, and we conclude that no best response to $\sigma^d_2$ exists.
\end{proof}

\subsection{Every non-trivial game has a strategy having a best response, but no computable best response}

We now present a computable strategy that has a best response, but no \emph{computable} best response. The game is split into periods consisting of one test stage and $K_r$ reward stages for some large enough integer $K_r$.

\begin{definition}\label{def:playertwo_good_strategy}
Let $G$ be non-trivial for \playerone{}, and let $K_r$ be the least integer satisfying 
$$\frac{1}{K_r+1}(u_1(D_1, C_2) + K_r M_1) > u_1(D_1, D_2).$$ Define $\sigma^e_{2}$ to be the strategy for \playertwo{} that, given a~finite history $h^T \in \mathcal{H}^T_{G^\infty}$, plays the following action in stage $T+1$:
\begin{enumerate}[(i)]
    \item If, for any $t$ satisfying $0 < K_r t < T$, either
    \begin{align}
    \label{eq:sigmae2cond1}
        h^T_1[K_rt+1] \neq C_1\ \&\ (t \in \insepsetA_T),
    \end{align} or 
    \begin{align}
    \label{eq:sigmae2cond2}
        h^T_1[K_rt+1] \neq D_1\ \&\ (t \in \insepsetB_T),
    \end{align} play $D_2$. 
    \item Otherwise, play $C_2$.\qed
\end{enumerate}
\end{definition}

We first prove that the strategy  $\sigma^e_{2}$ has a best response: 
\begin{definition}
\label{def:sigmae1}
Define $\sigma^e_{1}$ to be the strategy for \playerone{} that, given a finite history $h^T \in \mathcal{H}^T_{G^\infty}$, plays the following action in stage $T+1:$
(i)   If there exists some $t \in \insepsetB$ such that $T=K_r t$,
    then play $D_1$; (ii) otherwise, play $C_1$.
\end{definition}

\begin{lemma}
\label{lemma:inconsistent_avg_comp}
$\sigma^e_2$ is a computable strategy, and 
if $G$ is non-trivial for \playerone{}, then $\sigma^e_2$ has a best response, but no best response to  $\sigma^e_2$ is a computable strategy.
\end{lemma}
\begin{proof}
By Remark \ref{rem:setsABthenAnBn}, there is a Turing machine
that, on input $T$ will output the G{\"o}del number, $k$, of a Turing machine deciding
$\insepsetA_T$ (and similarly for $\insepsetB_T$); by using a universal Turing machine to simulate $T_k$, 
it is clearly decidable whether, for any $t$ such that $0 < K_rt < T$,
we have $t \in \insepsetA_T$, respectively $t \in \insepsetB_T$; and clearly, it is directly checkable
by a simple lookup in the history $h^T$, whether   $h^T_1[K_rt+1] \neq C_1
$, respectively $h^T_1[K_rt+1] \neq D_1$. Hence, $\sigma^e_2$ is a computable strategy. 

Let $h^\infty$ be the path of play induced by the strategy profile $s=(\sigma^e_1, \sigma^e_2)$, and let $h^\infty_i[T]$ be the action played by \playeri{} in stage $T$. \playertwo{} starts by playing $C_2$, and plays $D_2$ only if condition (\ref{eq:sigmae2cond1}) or (\ref{eq:sigmae2cond2}) in Definition \ref{def:playertwo_good_strategy} is satisfied for some $t, T \in \mathbb{N}$. Condition (\ref{eq:sigmae2cond1}) implies $t \in \insepsetA$, in which case $h_1^T[K_rt+1] = C_1$ by the definition of $\sigma^e_1$. Thus, $h_1^T[K_rt+1] \neq C_1$ is not satisfied, and the symmetric argument applies to (\ref{eq:sigmae2cond2}), so \playerone{} always plays $D_2$. \playerone{}'s payoff in every test stage is at least $u_1(D_1, C_2)$, and in every reward stage equals $u_1(C_1, C_2)$, so the limit-of-means payoff of \playerone{} is:
$$
\lompayoffone{\sigma^e_1}{\sigma^e_2} = \liminf_{T \to \infty}{\frac{1}{T}\sum_{i=1}^T{u_1(h^\infty_1[i], h^\infty_2[i])}} 
\geq
\frac{1}{K_r+1}(u_1(D_1, C_2) + K_r M_1) 
> u_1(D_1, D_2).
$$

Let $\bar{s}_1$ be any strategy for \playerone{}, define $\bar{s} = (\bar{s}_1, \sigma^e_2)$, and let $\bar{h}^\infty$ be the path of play induced by the strategy profile $\bar{s}$. By definition
of $\sigma^e_2$, in each stage, \playertwo{} either plays $C_2$ or $D_2$. There are thus two possibilities:
\begin{itemize}
    \item \playertwo{} always plays $C_2$, that is, for all $T \in \mathbb{N}$, $\bar{h}^\infty_2[T] = C_2$. 
    Assume, for contradiction, that $\bar{s}_1$ is a strictly better response to $\sigma^e_2$ than $\sigma^e_1$, that is, that $\lompayoffone{\bar{s}_1}{ \sigma^e_2} > \lompayoffone{\sigma^e_1}{\sigma^e_2}$. Then, for some stage $T \in \mathbb{N}$, $u_1(\bar{h}^\infty_1[T], C_2) > u_1(h^\infty_1[T], C_2)$, and because $u_1(C_1, C_2) = M_1$, we have $h^\infty_1[T] = D_1$ and $\bar{h}^\infty_1[T] \neq D_1$. But if $h^\infty_1[T] = D_1$, then $T = K_rt + 1$ for some $t \in \insepsetB$ and by Remark \ref{rem:setsABthenAnBn} there is then some $m \in \mathbb{N}$ such that $t \in \insepsetB_m$. Because $\bar{h}^\infty_1[K_rt + 1] \neq D_1$ and $t \in \insepsetB_m$, the definition of $\sigma^e_2$ yields that $\bar{h}^\infty_2[m+1] = D_2$, contradicting that \playertwo{} always plays $C_2$. Hence, $\lompayoffone{\bar{s}_1}{ \sigma^e_2} \leq \lompayoffone{\sigma^e_1}{\sigma^e_2}$.
    \item \playertwo{} plays $D_2$ in some stage $T+1$, that is, $\bar{h}^\infty_2[T+1] = D_2$. By definition of $\sigma^e_2$, there is some $t \in \mathbb{N}$ such that $K_rt < T$, and either condition (\ref{eq:sigmae2cond1}), or condition (\ref{eq:sigmae2cond2}), in Definition \ref{def:playertwo_good_strategy} is satisfied. Therefore, \playertwo{} continues playing $D_2$ forever, and hence:
    $$
    \lompayoffone{\bar{s}_1}{\sigma^e_2} = \liminf_{T \to \infty}{\frac{1}{T}\sum_{i=1}^T{u_1(\bar{h}^\infty_1[i], \bar{h}^\infty_2[i])}} \leq
    u_1(D_1, D_2) < \lompayoffone{\sigma^e_1}{\sigma^e_2}
    $$
    \end{itemize}

By the above, $\sigma^e_1$ is a best response to $\sigma^e_2$. To prove that no computable best response exists, assume, for contradiction, that there is a computable best response $\bar{s}_1$ to  $\sigma^e_2$. By the above analysis, we know that if \playertwo{} ever plays $D_2$, then $\lompayoffone{\sigma^e_1}{\sigma^e_2} > \lompayoffone{\bar{s}_1}{ \sigma^e_2}$. As $\bar{s}_1$ is a best response, $\bar{s}_1$ must thus ensure that \playertwo{} plays $C_2$ at every stage in the game. 
Hence, if $t \in \insepsetA$, we must have  $\bar{h}^\infty_1[K_rt+1] = C_1$,
and if $t \in \insepsetB$, we must have $\bar{h}^\infty_1[K_rt+1] = D_1$ (as otherwise,
$\sigma^e_2$ will play $C_2$).
 As $\bar{s}_1$ was assumed to be computable, there is a Turing machine $\textrm{TM}_{\bar{s}_1}$ computing $\bar{s}_1$. But then we can construct a Turing machine $T_k$ that uses $\textrm{TM}_{\bar{s}_1}$ as a subroutine and accepts if $\bar{h}^\infty_1[K_rt+1] = C_1$,
rejects if $\bar{h}^\infty_1[K_rt+1] = D_1$, and rejects if $\bar{h}^\infty_1[K_rt+1] \notin \{C_1,D_1\}$. But then $T_k$ halts on all inputs and decides the language
 $C = \{n \in \mathbb{N} : T_k \textrm{ accepts } n\}$; but $\insepsetA \subseteq C$ and $\insepsetB \cap C = \emptyset$, whence $C$ is a decidable set separating $\insepsetA$ and $\insepsetB$, 
 contradicting Proposition \ref{prop:sets_AB}. Hence, there is no computable best response to
 $\sigma^e_2$, as desired.
\end{proof}

\subsection{A complete characterization}

We now have our first main result:

\begin{theorem}\label{the:LOM_naive_the}
Let $G$ be a 2-player normal-form game. The following are equivalent under limit-of-means payoff
in $G^\infty$:
\begin{enumerate}[\textbf{(\alph*)}]
     \item \label{cond:characterization_non-trivial} $G$ is non-trivial for \playerone{}.
 \item \label{cond_characterization_no_br} There is a strategy for \playertwo{} in $G^\infty$
 that has no best response.
 \item \label{cond_characterization_no_comp_br} There is a computable strategy for \playertwo{} in $G^\infty$ that has no best response.
    \item \label{cond:characterization_br_avg}
    There is a strategy profile $s=(s_1,s_2)$ in $G^\infty$ satisfying
        \begin{enumerate}[(1)]
        \item \label{cond:characterization_br_avg_1} $s_1$ is a best response to $s_2$,
        \item \label{cond:characterization_br_avg_2} $s_2$ is computable,
        \item \label{cond:characterization_br_avg_3} $s_2$ does not have a computable best response.
        \end{enumerate}
   
\end{enumerate}
\end{theorem}
\begin{proof}
We prove that \ref{cond:characterization_non-trivial} is equivalent to 
\ref{cond_characterization_no_br}, \ref{cond_characterization_no_comp_br},
and \ref{cond:characterization_br_avg}.
If \ref{cond:characterization_non-trivial} holds, then
Lemma \ref{lem:non_trivial_no_best_response} yields existence of a computable strategy
in $G^\infty$ that has no best response, whence \ref{cond_characterization_no_comp_br}
and, a fortiori, \ref{cond_characterization_no_br} hold. Similarly,
if \ref{cond:characterization_non-trivial} holds, Lemma \ref{lemma:inconsistent_avg_comp}
yields that \ref{cond:characterization_br_avg} holds. 
If either of \ref{cond_characterization_no_br} or \ref{cond_characterization_no_comp_br} hold, it follows from Lemma \ref{lemma:avg_trivial_br} that $G$ is non-trivial for \playerone{}, hence
that \ref{cond:characterization_non-trivial} holds. Similarly, if \ref{cond:characterization_non-trivial} holds, the strategy $s_2$ is computable,
but has no computable best response, whence Lemma \ref{lemma:avg_trivial_br} 
yields that $G$ is non-trivial for \playerone{}, and thus that \ref{cond:characterization_non-trivial}
holds.
\end{proof}



\input{check_triviality}

\section{Intermezzo: Folk theorems}

Folk theorems characterize the payoff profiles that are achievable under equilibria in different settings, depending on how the payoff is computed or which kind of equilibria we are interested in. The conclusion of all folk theorems is approximately the same -- every payoff profile satisfying two minimal requirements is achievable under a Nash equilibrium. First, individual rationality, demands at least the obvious minimal payoff (the \emph{minmax} payoff) for every player and second, feasibility, ensures that the payoffs in the repeated game can be combined from the stage game payoffs. The proofs of folk theorems are usually constructive and provide us with actual strategy profiles that lead to given payoff profiles--we restate well-known folk theorems below with explicit assertions of the existence of computable equilibria (we stress that the proofs and proof ideas are not ours, but are already well-known). We shall use the folk theorems in Section \ref{sec:NE_SPE}.

\begin{definition}
Let $G = (N, A, u)$ be a normal-form game and let $\upsilon=(\upsilon_1, \dots, \upsilon_n) \in \mathbb{Q}^n$ be a payoff profile. 
\begin{enumerate}[(1)]
    \item $\upsilon$ is said to be \emph{individually rational for} \playeri{} if $\upsilon_i \geq \min\limits_{a_{-i} \in A_{-i}}\max\limits_{a_{i} \in A_{i}} u_i(a_i, a_{-i})$. Furthermore,
    $\upsilon$ is said to be \emph{strictly} individually rational for \playeri{} if the inequality is strict. $\upsilon$ is said to be individually rational if it is individually rational for all players.
    
    \item $\upsilon$ is said to be (rationally) \emph{feasible} if there exists a non-negative vector $\alpha \in \mathbb{Q}^{|A|}$ satisfying \\$\sum_{a \in A}{\alpha_a} = 1$ and $\forall i \in N: \upsilon_i = \sum_{a \in A}{\alpha_a u_i(a)}$\footnote{In some places in the literature, $\upsilon$ is called feasible if the vector $\alpha$ is merely required to be an element of $\mathbb{R}^{\vert A \vert}$
    instead of $\mathbb{Q}^{\vert A \vert}$--in which case the set of feasible payoff profiles $\upsilon$ is exactly the elements in 
    the convex hull of the set of payoff profiles of the stage game. We prefer to keep $\alpha \in \mathbb{Q}^{\vert A \vert}$ as it yields
    a cleaner statement of the constructive parts of folk theorems.}.
\end{enumerate}
\end{definition}

\begin{remark}\label{rem:sheesh_cats}
$\upsilon \in \mathbb{Q}^{\vert A \vert}$ is feasible if{f} it is in the convex hull of the payoff profiles of $G$: If $\upsilon \in \mathbb{Q}^{\vert A \vert}$ is in the convex hull of the payoff profiles of $G$, there is a vector
$\beta \in \mathbb{R}^{\vert A \vert}$ with non-negative components such that
$\sum_{a \in A} \beta_a = 1$ and $\forall i \in N : \upsilon_i = \sum_{a \in A} \beta_a u_i (a)$,
and as the components of $\beta$ are non-negative and $\upsilon_i$ and all $u_i(a)$ are rational,
then the components of $\beta$ must be rational. Conversely, if $\upsilon$ is feasible,
then $\upsilon$ is obviously in the convex hull of the point set $\cup_{a \in A}\{u(a)\}$. 

Observe that \emph{any} payoff of a strategy profile of $G^\infty$ is in the convex hull
of the payoff profiles of $G$ (this follows from direct inspection of the definition of limit-of-means payoff)---essentially the weights of the payoff profiles correspond to the frequency with which the payoff profiles occur.
\end{remark}

The following theorem is adapted from Aumann \cite{bib:aumann1981}:

\begin{theorem}[Folk Theorem--Nash Equilibria]
\label{theorem:folk_avg_ne}
Let $G = (N,A,u)$ be a normal-form game.  
\begin{enumerate}[(1)]
    \item If $\upsilon \in \mathbb{R}^N$ is a payoff profile under a Nash equilibrium in $G^\infty$ then $\upsilon$ is individually rational. 
    \item If $\upsilon \in \mathbb{Q}^N$ is feasible and individually rational then there is a Nash equilibrium $s$ in $G^\infty$ leading to the payoff profile $\upsilon$ such that every strategy in $s$ is computable.
\end{enumerate}
\end{theorem}

Theorem $\ref{theorem:folk_avg_ne}$ describes a set of payoff profiles that are achievable under a Nash equilibrium. Aumann and Shapley \cite{Aumann1994LongTermC} (see also \cite[Prop.\ 146.2]{Osborne1994}) prove a similar result for subgame-perfect equilibria. As every subgame-perfect equilibrium is also a Nash equilibrium, Theorem $\ref{theorem:folk_avg_ne}$ implies that every payoff profile under a subgame-perfect equilibrium is individually rational. Surprisingly, the sufficient condition for the existence of subgame-perfect equilibria with a given payoff profile is not stricter than for Nash equilibria:

\begin{theorem}[Folk Theorem--subgame-perfect equilibria]
\label{theorem:folk_avg_spe}
Let $G = (N,A,u)$ be a normal-form game, and let $\upsilon \in \mathbb{Q}^{N}$ be a feasible and individually rational payoff profile.
There is a subgame-perfect equilibrium $s$ in $G^\infty$ leading to the payoff profile $\upsilon$ in every subgame such that every strategy in $s$ is computable.
\end{theorem}

\section{Computability under Nash and subgame-perfect equilibria}\label{sec:NE_SPE}

Theorem \ref{the:LOM_naive_the} shows that any game satisfying some simple conditions has a computable strategy without a computable best response. However, strategy profiles are generally not of much interest unless they are Nash equilibria, or subgame-perfect equilibria. We treat these cases in the present section.
We first prove two auxiliary lemmas before obtaining a complete characterization at the end of the section.

\subsection{Nash equilibria}

We first treat Nash equilibria. Observe that if $s = (s_1,s_2)$ is a Nash equilibrium of $G^\infty$, then $s_1$ is a best response to $s_2$.

\begin{lemma}\label{lem:comp_avg_ne}
Let $G = (N,A,u)$ be a 2-player normal-form game. The following are equivalent:
\begin{enumerate}[\textbf{(\alph*)}]
    \item \label{cond:le_foo} There is a strategy profile $s=(s_1,s_2)$ in $G^\infty$ satisfying
        \begin{enumerate}[(1)]
        \item \label{cond:characterization_ne_avg_1} $s$ is a Nash equilibrium of $G^\infty$,
        \item \label{cond:characterization_ne_avg_2} $s_2$ is computable,
        \item \label{cond:characterization_ne_avg_3} $s_2$ does not have a computable best response.
        \end{enumerate}
    \item \label{ineq:gt_minmax_ne} $\vert A_1 \vert \geq 2,$ and there is a Nash equilibrium $s'$ of $G^\infty$ that is strictly individually rational for \playerone{}, that is, $s'$ satisfies:
    \begin{align*}
        \lompayoffname_1(s') > \min\limits_{a_{2} \in A_{2}}\max\limits_{a_{1} \in A_{1}} u_1(a_1, a_2).
    \end{align*}
\end{enumerate}
\end{lemma}

\begin{proof}
If $G^\infty$ does not have a Nash equilibrium then the equivalence is obvious. If $G^\infty$
does have a strategy profile $s'$ that is a Nash equilibrium,
 Theorem $\ref{theorem:folk_avg_ne}$ yields that $s'$ leads to an individually rational payoff profile,
 and thus in particular $\lompayoffname_1(s') \geq \min\limits_{a_{2} \in A_{2}}\max\limits_{a_{1} \in A_{1}} u_1(a_1, a_2)$. 

We first prove that $\neg \ref{ineq:gt_minmax_ne} \Rightarrow \neg \ref{cond:le_foo}$.
Assume that \ref{ineq:gt_minmax_ne} does not hold; if $\vert A_1 \vert = 1$, then there is exactly
one strategy for \playerone{}, namely the one always playing the single action available to \playerone{},
and is obviously both computable and a best response to any strategy of \playertwo{},
and hence \ref{cond:le_foo} does not hold. If $\vert A_1 \vert \geq 2$ and \ref{ineq:gt_minmax_ne}
does not hold,
then every Nash equilibrium $s'$ of $G^\infty$ leads to a payoff profile $\lompayoffname(s')$ satisfying $\lompayoffname_1(s') = \min\limits_{a_{2} \in A_{2}}\max\limits_{a_{1} \in A_{1}} u_1(a_1, a_2)$.
 If there is no Nash equilibrium $s = (s_1,s_2)$ such that $s_2$ is computable, it follows immediately
 that \ref{cond:le_foo} does not hold. So, assume that $s = (s_1,s_2)$ is a Nash equilibrium such that $s_2$ is computable. Define $\bar{s}_1$ to be the strategy of Player $1$ that in stage $T+1$, given a finite history $h^T \in \mathcal{H}^T_{G^\infty}$, first simulates $s_2$ to obtain $\bar{a}_2 = s_2(h^T)$, and then
plays any action $\bar{a}_1$ satisfying $u_1(\bar{a}_1, \bar{a}_2) = M_1(\bar{a}_2) \geq \min\limits_{a_{2} \in A_{2}}\max\limits_{a_{1} \in A_{1}} u_1(a_1, a_2)$. 
As $\lompayoffone{\bar{s}_1}{s_2} \geq \min\limits_{a_{2} \in A_{2}}\max\limits_{a_{1} \in A_{1}} u_1(a_1, a_2) = \lompayoffone{s_1}{s_2}$ and $s_1$ is a best response to $s_2$, $\bar{s}_1$ is also a best response to $s_2$. Moreover, $\bar{s}_1$ is clearly a computable strategy because $s_2$ is. Thus,
\ref{cond:characterization_ne_avg_3} does not hold, and $\neg \ref{ineq:gt_minmax_ne} \Rightarrow \neg \ref{cond:le_foo}$ follows.

We now prove $\ref{ineq:gt_minmax_ne} \Rightarrow \ref{cond:le_foo}$. Assume that $\ref{ineq:gt_minmax_ne}$ holds; by Theorem $\ref{theorem:folk_avg_ne}$, there is a payoff profile that is feasible and strictly individually rational for \playerone{} leading to a Nash equilibrium $s=(s_1, s_2)$ where $s_1$ and $s_2$ are computable, and $\lompayoffname(s) = \lompayoffname(s')$. We will modify $s$ so that no best response to \playertwo{}'s strategy is computable. The modification adds a test procedure to the stages of $s_2$ that are a power of 2. \playertwo{} will verify if \playerone{} played the correct action in all previous test stages, and if this test is passed, both players will pretend that they played according to $s$ in test stages when deciding to play the next action.
Formally, fix any $C_1, D_1 \in A_1$ such that $C_1 \neq D_1$ and define $D_2 \in A_2$ as $D_2 = \argmin\limits_{a_2 \in A_2}\max\limits_{a_1 \in A_1}{u_1(a_1, a_2)},$ hence $M_1(D_2) < \lompayoffname_1(s)$. Let $h^{(s)}$ be the path of play obtained by playing $s.$ For a finite history $h^T \in \mathcal{H}^T_{G^\infty},$ define $\textrm{fix}(h^T) \in \mathcal{H}^T_{G^\infty}$ by $\textrm{fix}(h^T)[t] = h^{(s)}[t]$ if $t=2^i$ for some $i \in \mathbb{N}$, and $\textrm{fix}(h^T)[t] = h^T[t]$ otherwise.
Now, define $\bar{s}_2$ to be the strategy for \playertwo{} that, given a~finite history $h^T \in \mathcal{H}^T_{G^\infty},$  plays the following action in stage $T+1$:
\begin{enumerate}
    \item If for any $t$ satisfying $0 < 2^t \leq T$, either
    \begin{align*}
        h^T_1[2^t] \neq C_1\ \&\ (t \in \insepsetA_T),
    \end{align*} or 
    \begin{align*}
        h^T_1[2^t] \neq D_1\ \&\ (t \in \insepsetB_T),
    \end{align*} play $D_2$. 
    \item Otherwise, play $s_2(\textrm{fix}(h^T))$.
\end{enumerate}
Now, define $\bar{s}_1$ to be the strategy for \playerone{} that, given a~finite history $h^T \in \mathcal{H}^T_{G^\infty},$ plays the following in stage $T+1$: 
\begin{enumerate}
    \item If $T+1=2^t$ for some $t \in \insepsetB$, play $D_1$.
    \item If $T+1=2^t$ for some $t \in \insepsetA$, play $C_1$.
    \item Otherwise, play $s_1(\textrm{fix}(h^T))$.
\end{enumerate}
We claim that $\bar{s} = (\bar{s}_1, \bar{s}_2)$ is a Nash equilibrium of $G^\infty$, that $\bar{s}_2$
is a computable strategy, and that $\bar{s}_2$ does not have a computable best response, that is,
all three conditions of \ref{cond:le_foo} are satisfied. First observe that  $\bar{s}_2$ is a computable strategy: As $s_2$ is computable, a Turing machine that computes
it can be used as a subroutine by a Turing machine $\textrm{TM}$ that, by Remark \ref{rem:setsABthenAnBn}, on input $T$ can generate the G{\"o}del numbers of
Turing machines deciding $\insepsetA_T$ and $\insepsetB_T$, and subsequently simulate these on input $t$ using a universal Turing machine as a subroutine.

Next, we prove that $\bar{s}$ is a Nash equilibrium. Observe that the payoffs in test stages $T$ satisfy:
$$
\frac{1}{T} \sum_{i=1}^{\lfloor \log_2{T} \rfloor}{u_1(h_{\bar{s}}^\infty[2^i]}) = 
\frac{1}{T} \sum_{i=1}^{\lfloor \log_2{T} \rfloor}{u_1(h_{s}^\infty[2^i]}) = 
\frac{O(\log_2{T})}{T} 
$$
and the paths of play $h_s^\infty$ and $h_{\bar{s}}^\infty$ are identical at non-test-stages, whence at any test
stage $T$:
$$
 \sum_{i=1}^T{u_1(h^\infty_{\bar{s}}[i])} -
\sum_{i=1}^{\lfloor \log_2{T} \rfloor}{u_1(h_{\bar{s}}^\infty[2^i])}  =   \sum_{i=1}^T{u_1(h^\infty_{s}[i])} - 
\sum_{i=1}^{\lfloor \log_2{T} \rfloor}{u_1(h_{s}^\infty[2^i])}
$$

Hence, \playerone{}'s payoff is:

\begin{align*}
\lompayoffname_1(\bar{s}) &= 
\liminf_{T \to \infty}{\frac{1}{T} \sum_{i=1}^{T}{u_1(h_{\bar{s}}^\infty[i]})} \\  
&= \liminf_{T \to \infty}{\frac{1}{T} \left( \sum_{i=1}^{\lfloor \log_2{T} \rfloor}{u_1(h_{\bar{s}}^\infty[2^i])} + \sum_{i=1}^T{u_1(h^\infty_{\bar{s}}[i])} - 
\sum_{i=1}^{\lfloor \log_2{T} \rfloor}{u_1(h_{\bar{s}}^\infty[2^i])} \right)} \\
&\geq \liminf_{T \to \infty}{\frac{1}{T} \sum_{i=1}^{\lfloor \log_2{T} \rfloor}{u_1(h_{\bar{s}}^\infty[2^i])} } + \liminf_{T \to \infty}{\frac{1}{T} \left( \sum_{i=1}^T{u_1(h^\infty_{s}[i])} - 
\sum_{i=1}^{\lfloor \log_2{T} \rfloor}{u_1(h_{s}^\infty[2^i])} \right)} \\
&= \liminf_{T \to \infty}{\frac{O(\log_2{T})}{T}}  + 
\liminf_{T \to \infty}{\frac{1}{T}
\left(\sum_{i=1}^T u_1(h^\infty_s[i]) - O(\log_2 T)\right)} \\   
 &= \liminf_{T \to \infty}{\frac{1}{T}\sum_{i=1}^T u_1(h^\infty_s[i])}
 = \lompayoffname_1(s) > M_1(D_2) 
\end{align*}

By definition of $\bar{s}_2$, if \playerone{} ever deviates from the strategy $\bar{s}_1$ in a test stage ($T = 2^t$), \playerone{} will obtain the limit-of-means payoff $M_1(D_2) < \lompayoffname_1(\bar{s})$, making the deviation unprofitable. If \playerone{} deviates from $\bar{s}_1$ (and hence also $s_1$) at a non-test-stage, \playerone{} cannot 
obtain strictly greater payoff than $\lompayoffname_1(\bar{s}) = \lompayoffname_1(s)$ because $s$ is a Nash equilibrium. The paths of play obtained by playing $s$ and $\bar{s}$ are identical outside of test stages, so by a symmetric argument, \playertwo{}'s payoff $\lompayoffname_2(\bar{s}) = \lompayoffname_2(s).$ If \playertwo{} ever deviates, \playerone{} punishes them by switching to the minmax against \playertwo{} forever. Because $\lompayoffname(s)$ is individually rational, this deviation also cannot yield strictly greater payoff, and hence $\bar{s}$ is a Nash equilibrium. 

It remains to prove that $\bar{s}_2$ has no computable best response. 
Observe that any best response $\bar{s}_1'$ to $\bar{s}_2$ cannot deviate from $\bar{s}_1$ in any test stage, because otherwise \playertwo{} would play $D_2$ forever, resulting
in \playerone{} obtaining payoff at most $M_1(D_2) < \lompayoffname_1(\bar{s})$. Assume, for contradiction, that $\bar{s}_1'$ were computable; then, there is a Turing machine $\textrm{TM}_1'$ computing $\bar{s}_1'$. As $\bar{s}_2$ is computable, let $\textrm{TM}_2$ be a Turing machine computing $\bar{s}_2$, and let $\textrm{TM}_g$ be a Turing machine that, on input $t \in \mathbb{N}$, first simulates both $\textrm{TM}_1'$ and $\textrm{TM}_2$ $2^t$ times to obtain the action profile $h^{2^t}[2^t]$ induced by $(\bar{s}_1', \bar{s}_2)$. Then, $\textrm{TM}_g$ accepts if $h_1^{2^t}[2^t] = C_1$, and rejects otherwise. Observe that $\textrm{TM}_g$ then decides a language $C$ such that $\insepsetA \subseteq C$
and $\insepsetB \cap C = \emptyset$, contradicting the fact
that $\insepsetA$ and $\insepsetB$ are recursively inseparable by Proposition \ref{prop:sets_AB}.
\end{proof}

\subsection{Subgame-perfect equilibria}

As Theorem \ref{theorem:folk_avg_spe} states that focusing on subgame-perfect equilibria does not narrow the set of payoff profiles compared to Nash equilibria, it should be
no surprise that the below lemma for subgame-perfect equilibria requires the same conditions as Lemma \ref{lem:comp_avg_ne}

\begin{lemma}
\label{lem:comp_avg_spe}
Let $G$ be a 2-player normal-form game. The following are equivalent:
\begin{enumerate}[\textbf{(\alph*)}]
    \item \label{cond:spe_avg} There is a strategy profile $s=(s_1,s_2)$ in $G^\infty$ satisfying
        \begin{enumerate}[(1)]
        \item \label{cond:characterization_spe_avg_1} $s$ is a subgame-perfect equilibrium of $G^\infty$,
        \item \label{cond:characterization_spe_avg_2} $s_2$ is computable,
        \item \label{cond:characterization_spe_avg_3} $s_2$ does not have a computable best response.
        \end{enumerate}
    \item \label{ineq:gt_minmax_spe} $\vert A_1 \vert \geq 2,$ and there is a Nash equilibrium $s'$ of $G^\infty$ that is strictly individually rational for \playerone{}, that is, $s'$ satisfies:
    \begin{align}
        \lompayoffname_1(s') > \min\limits_{a_{2} \in A_{2}}\max\limits_{a_{1} \in A_{1}} u_1(a_1, a_2).
    \end{align}
    \end{enumerate}
\end{lemma}
\begin{proof}
We first prove \ref{cond:spe_avg} $\Rightarrow$ \ref{ineq:gt_minmax_spe}. Assume
that \ref{cond:spe_avg} holds. If $\vert A_1 \vert = 1$, then there is a single strategy for \playerone{},
namely the one always playing the single action in $A_1$; clearly, this strategy is computable,
whence \ref{cond:spe_avg} could not hold, contradicting the assumption, and we thus conclude that $\vert A \vert \geq 2$; furthermore,
a strategy profile that is a subgame-perfect equilibrium is also a Nash equilibrium, and by Lemma \ref{lem:comp_avg_ne} we thus conclude that \ref{ineq:gt_minmax_spe} holds. 

The proof of \ref{ineq:gt_minmax_spe} $\Rightarrow$ \ref{cond:spe_avg} proceeds in
 the same fashion as the proof of Lemma \ref{lem:comp_avg_ne}, with some modifications
 to accomodate subgame-perfect equilibria.
Assume that \ref{ineq:gt_minmax_spe} holds, and let $s'$ be a Nash equilibrium of $G^\infty$ such that $ \lompayoffname_1(s') > \min\limits_{a_{2} \in A_{2}}\max\limits_{a_{1} \in A_{1}} u_1(a_1, a_2)$. By Theorem \ref{theorem:folk_avg_spe}, there is a subgame-perfect equilibrium $s = (s_1, s_2)$ where both $s_1$ and $s_2$ are computable and $\lompayoffname(s) = \lompayoffname(s')$. We will modify $s$ as in the proof of Lemma \ref{lem:comp_avg_ne}, except that we allow only \emph{finite} punishments (to ensure a subgame-perfect equilibrium). If a deviation in a test stage is detected in stage $T_D$, then \playertwo{} plays the minmax against \playerone{} for the next $T^2_D - T_D$ stages (we will colloquially call this a \emph{punishment phase} of the game). 

Fix any $C_1, D_1 \in A_1$ such that $C_1 \neq D_1$, and let $D_2 \in A_2$ be a minmax action against \playerone{}, that is, $M_1(D_2) = u_1(D_2) =\min\limits_{a_{2} \in A_{2}}\max\limits_{a_{1} \in A_{1}} u_1(a_1, a_2)$; observe that $M_1(D_2) < \lompayoffname_1(s)$. For $n \in \insepsetA \cup \insepsetB$, define $\textrm{detect}(n)$ to be the least $m \in \mathbb{N}$ such that $n \in \insepsetA_m \cup \insepsetB_m$. For a finite history $h^T \in \mathcal{H}^T_{G^\infty},$ define $\textrm{fix}(h^T) \in \mathcal{H}^T_{G^\infty}$ by $\textrm{fix}(h^T)[t] = h^{(s)}[t]$ if $t=2^i$ for some $i \in \mathbb{N}$, and $\textrm{fix}(h^T)[t] = h^T[t]$ otherwise. Define $\bar{s}_2$ to be the strategy for \playertwo{} that, given a finite history $h^T \in \mathcal{H}^T_{G^\infty}$, plays the following action in stage $T+1$:
\begin{enumerate}
    \item \label{cond:fui} If for any $t$ satisfying $0 < 2^t \leq T$ and $\textrm{detect}(t) > \sqrt{T}$, either
    \begin{align*}
        h^T_1[2^t] \neq C_1\ \&\ (t \in \insepsetA_T),
    \end{align*} or 
    \begin{align*}
        h^T_1[2^t] \neq D_1\ \&\ (t \in \insepsetB_T),
    \end{align*}
    play $D_2$. 
    \item Otherwise, play $s_2(\textrm{fix}(h^T))$.
\end{enumerate}
Define $\bar{s}_1$ to be Player 1's strategy that, given a~finite history $h^T \in \mathcal{H}^T_{G^\infty},$ plays the following action in stage $T+1$: 
\begin{enumerate}
    \item If $T+1=2^t$ for some $t \in \insepsetB$, play $D_1$.
    \item if $T+1=2^t$ for some $t \in \insepsetA$, play $C_1$.
    \item Otherwise, play $s_1(\textrm{fix}(h^T))$.
\end{enumerate}

We claim that $\bar{s} = (\bar{s}_1, \bar{s}_2)$ is a subgame-perfect equilibrium,
that $\bar{s}_2$ is computable, and that no best response to $\bar{s}_2$ is a computable strategy.

We first prove that $\bar{s}_2$ is a computable strategy: As $s_2$ is computable, a Turing machine that computes
it can be used as a subroutine by a Turing machine $\textrm{TM}$ that, by Remark \ref{rem:setsABthenAnBn}, on input $T$ can generate the G{\"o}del numbers of
Turing machines deciding $\insepsetA_T$ and $\insepsetB_T$, and subsequently simulate these on inputs on the form $2^t$ using a universal Turing machine as a subroutine. Furthermore, even though the function $\textrm{detect}$ is not directly computable, it is decidable whether
$\textrm{detect}(t) > \sqrt{T}$ as it suffices to generate all elements of the
sets $\insepsetA_i$ and $\insepsetB_i$ with $i \in \{1,\ldots,\lfloor \sqrt{T} \rfloor\}$,
which can be done by Remark \ref{rem:setsABthenAnBn}.

We proceed to prove that $\bar{s}$ is a subgame-perfect equilibrium. Observe that, by the definition of $\bar{s}$, $\lompayoffname_1(\bar{s}) = \lompayoffname_1(s) > M_1(D_2)$ and $\lompayoffname_2(\bar{s}) = \lompayoffname_2(s).$  Consider a finite history $h^T \in \mathcal{H}_{G^\infty}^T$ and a subgame $(G^\infty, h^T)$. If both players play $\bar{s}$ in $(G^\infty, h^T)$, they obtain the payoff profile $\lompayoffname(\bar{s})$, as any punishment phase of a deviation in $h^T$ lasts at most until stage 
$$
(\max\{\textrm{detect}(1), \textrm{detect}(2),\ldots, \textrm{detect} (\lfloor \log_2 T \rfloor)\})^2
$$ 
 Because $s$ is a subgame-perfect equilibrium, any unilateral deviation by \playertwo{} cannot result
 in strictly greater payoff for \playertwo{}. If \playerone{} deviates for only a finite number of stages, let $T_X$ be the stage of the last such a deviation. Starting from stage 
$$
1 + (\max\{\textrm{detect}(1),\textrm{detect}(2),\ldots,\textrm{detect}(\lfloor\log_2(T_X)\rfloor)\})^2 
$$ 
\noindent the path of play returns to the path of play determined by $\bar{s}_1$ and $\bar{s}_2$, leading to the limit-of-means payoff $\lompayoffname(\bar{s})$ again. Now, if \playerone{} deviates in infinitely many stages, let $T_0, T_1, T_2, \dots$ be the sequence of stages when \playerone{}'s deviation is first detected by \playertwo{} (i.e., when case (\ref{cond:fui}) in the definition of $\bar{s}_2$ applies after a period of playing $s_1(\textrm{fix}(h^T))$). 
The payoff at the end of the $n$th punishment phase
(that is, in stage $T_n + (T^2_n - T_n) = T^2_n)$ is
at most
$$
\frac{1}{T^2_n}\left(T_n \max_{a \in A}\{ u_1(a)\} + (T^2_n - T_n) M_1(D_1) \right)
\leq M_1(D_1) + O(1/n) 
$$
As the payoff of the infinitely repeated game
is the limit inferior of the payoffs after finitely many stages, \playerone{}'s payoff when performing infinitely many deviations is thus in particular at most
$$
\liminf_{n\rightarrow\infty}{\left(\min\limits_{a_{2} \in A_{2}}\max\limits_{a_{1} \in A_{1}} u_1(a_1, a_2) + O(1/n)\right)}
\leq M_1(D_1) \leq \lompayoffname_1(\bar{s})
$$
whence no deviation of \playerone{} can result in strictly greater payoff than
$\lompayoffname_1(\bar{s})$. Moreover, since $M_1(D_1) < \lompayoffname_1(\bar{s})$, any best response of \playerone{} can deviate only finitely many times. As no deviation of any player following any finite history would strictly increase their payoff, $\bar{s}$ is a subgame-perfect equilibrium.

We now prove that any best response to $\bar{s}_2$ is not computable.
%
Assume, for contradiction, that $\bar{s}_2$ has a computable best response $\bar{s}'_1$. By the previous argument, $\bar{s}'_1$ deviates from the prescribed path only finitely many times, and thus \emph{a fortiori} it deviates in test stages only finitely many times. 
Let $T_X=2^{t_x}$ be the last test stage where $\bar{s}'_1$ deviates. As $\bar{s}_2$ is computable, let $\textrm{TM}_2$ be a Turing machine computing $\bar{s}_2,$ and let $\textrm{TM}_1'$ be a Turing machine computing $\bar{s}_1'$. Let $\textrm{TM}_g$ be a Turing machine that, on input $t \in \mathbb{N}$ does the following:

\begin{itemize}

\item if $t \leq t_X$, $\textrm{TM}_g$ looks up in an array of length $t_X$ whether
$t \in \insepsetA$, and accepts if it is, and rejects otherwise.

\item If $t>t_X$, $\textrm{TM}_g$ first simulates both $\textrm{TM}_1'$ and $\textrm{TM}_2$ $2^t$ times to obtain the action profile $h^{2^t}[2^t]$ induced by $(\bar{s}_1', \bar{s}_2)$. Then, $\textrm{TM}_g$ accepts if $h_1^{2^t}[2^t]=C_1$, and rejects otherwise. 

\end{itemize}
Now, by construction, $\textrm{TM}_g$ halts on all inputs and decides a language $C$
such that $\insepsetA \subseteq C$ and $C \cap \insepsetB = \emptyset$,
contradicting the fact that $\insepsetA$ and $\insepsetB$ are recursively inseparable by Proposition \ref{prop:sets_AB}.

Thus, $\bar{s}$ satisfies all conditions $\ref{cond:characterization_spe_avg_1}$ - $\ref{cond:characterization_spe_avg_3}$, and thus \ref{ineq:gt_minmax_spe} holds, as desired.
\end{proof}

\subsection{A complete characterization of Nash and subgame-perfect equilibria}

We can now finally give a complete characterization of games
where a strategy for \playertwo{}--required to be part of a strategy profile that is either a Nash equilibrium or subgame-perfect equilibrium--has no computable best response:

\begin{theorem}\label{theorem:final_charac_lom}
Let $G$ be a 2-player normal-form game. The following are equivalent:

\begin{enumerate}[\textbf{(\alph*)}]
        \item \label{cond:le_foo} There is a strategy profile $s=(s_1,s_2)$ in $G^\infty$ satisfying
        \begin{enumerate}[(1)]
        \item \label{cond:characterization_ne_avg_1} $s$ is a Nash equilibrium of $G^\infty$,
        \item \label{cond:characterization_ne_avg_2} $s_2$ is computable,
        \item \label{cond:characterization_ne_avg_3} $s_2$ does not have a computable best response.
        \end{enumerate}
    
    \item There is a strategy profile $s=(s_1,s_2)$ in $G^\infty$ satisfying
        \begin{enumerate}[(1)]
        \item \label{cond:characterization_spe_avg_1} $s$ is a subgame-perfect equilibrium of $G^\infty$,
        \item \label{cond:characterization_spe_avg_2} $s_2$ is computable,
        \item \label{cond:characterization_spe_avg_3} $s_2$ does not have a computable best response.
        \end{enumerate}
    \item \label{cond:le_last} $\vert A_1 \vert \geq 2,$ and there is a Nash equilibrium $s'$ of $G^\infty$ that is strictly individually rational for \playerone{}, that is, $s'$ satisfies:
    \begin{align}
        \lompayoffname_1(s') > \min\limits_{a_{2} \in A_{2}}\max\limits_{a_{1} \in A_{1}} u_1(a_1, a_2).
    \end{align}
    \end{enumerate}
\end{theorem}

\begin{proof}
The result follows directly from Lemmas \ref{lem:comp_avg_ne} and \ref{lem:comp_avg_spe}.
\end{proof}

\input{convex_hull}

\section{Some examples of applying the results}

\input{examples}

\section{Conclusion and future work}

While we have provided a complete characterization of two-player games
with computable strategies without computable best responses in the case of limit-of-means payoff, there are other standard ways of 
defining the payoff--notably the discounted payoff
where sufficient conditions (not known to be necessary) exist \cite{bib:Nachbar1996}--and complete characterizations must be established for these as well. For the case of \emph{subrecursive} computation where strategies are computed by machines with strictly less extensional power than Turing machines,
some results are known, notably for time and space complexity classes (see, e.g.\ \cite{DBLP:conf/stoc/FortnowW94,bib:Chen:2015}), but it would be interesting to have a general result holding for all suitably well-behaved classes of (subrecursive) functions (e.g., classes axiomatizable as in \cite{Kozen1980}). In addition, repeated games are just a special case of sequential games that are usually represented in extensive form. Is it possible to apply the techniques used for infinitely repeated games to infinite extensive-form games?  

As similar computability problems can be investigated for games with imperfect information, or cooperative games, it would be interesting to derive complete characterizations of games with computable strategies without best responses in those settings; likewise, it would be interesting to investigate whether our results carry over to a setting where mixed strategies--as opposed to the pure strategies studied in this paper--are allowed.
Finally, all of the above can be investigated for games with more than two players, but this is likely to produce characterizations that are substantially harder to check than those in the present paper; for example, Nash and subgame-perfect equilibria are significantly harder to compute
for repeated games with more than two players under discounted payoff \cite{BORGS201034,DBLP:journals/geb/HalpernPS19}

\clearpage

%
%
%
%
%

\bibliographystyle{ACM-Reference-Format}
\bibliography{bibliography_Dargaj_Simonsen.bib}

\clearpage

\appendix
\section{Material omitted from the main text}

A full proof of Proposition \ref{prop:sets_AB} can be found in several publications (e.g., \cite{bib:Nachbar1996}). For completeness, we also give a full proof below using the notation of the present paper.

\begin{proof}[Proof of Proposition \ref{prop:sets_AB}]
Both $\insepsetA$ and $\insepsetB$ are clearly
recursively enumerable, and hence so is $\insepsetA \cup \insepsetB$. 

To prove that $A$ and $B$ are recursively inseparable, define $g: \mathbb{N} \to \{0, 1\}$ to be any (possibly partial) function satisfying:
\begin{enumerate}
    \item $g(n) = 1 \mbox{ if } n \in \insepsetA$,
    \item $g(n) = 0 \mbox{ if } n \in \insepsetB$.
\end{enumerate}
We claim that $g$ is not computable. Suppose, for contradiction, that $g$ were computable; then, let $k$ be the G{\"o}del number of a Turing machine 
such that $\phi_k = g$. Because $\halts{\phi_k(n)}$ for all 
$n \in \mathbb{N}$, we have $\halts{\phi_k(k)}$, and hence either $k \in \insepsetA$ or $k \in \insepsetB$. If $k \in \insepsetA$, we have $\phi_k(k) = 0$, but by definition we have $g(k) = 1$,
a contradiction. Otherwise, we have $k \in \insepsetB$,
and thus $\phi_k(k) \neq 0$;  but $g(k) = 0$, and we once again obtain a contradiction. Hence, $g$ is not computable. 
But if there were a decidable set $C$ such that $A \subseteq C$
and $B \cap C = \emptyset$, we can construct a Turing machine
with some G{\"o}del number $l$ such that 
$\phi_l(n) = 1$ if{f} $n \in C$ and $\phi_l(n) = 0$ otherwise.
But as $\insepsetA \cap \insepsetB = \emptyset$, 
$\phi_l$ then satisfies that $\phi_l(n) = 1$ of $n \in \insepsetA$ and $\phi_l(n) = 0$ if $n \in \insepsetB$
which contradicts the above observation that no such function is computable.
Consequently, $\insepsetA$ and $\insepsetB$ are recursively inseparable.

Observe that if $\insepsetA$ were decidable, then the
fact that $\insepsetA \cap \insepsetB = \emptyset$ implies that
$\insepsetA$ is a decidable set separating $\insepsetA$ and
$\insepsetB$, contradicting recursive inseparability of the two sets. The proof of undecidability of $\insepsetB$ is symmetric.
Finally, note that $\insepsetA \cup \insepsetB = \{n \in \mathbb{N} : \halts{\phi_n(n)}\} = \emptyset'$, and hence
$\insepsetA \cup \insepsetB$ is undecidable.
\end{proof}

Below is a full proof--adapted from a proof by Aumann \cite{bib:aumann1981} but using the notation from the present paper--of Theorem \ref{theorem:folk_avg_ne}:

\begin{proof}[Proof of Theorem \ref{theorem:folk_avg_ne}]
(1) Assume, for contradiction, that there exists a Nash equilibrium $s=(s_1, \dots, s_n)$ of $G^\infty$ with such that $\upsilon_i(s) = \upsilon_i$ for all $i \in N$, but that there is some $i \in N$ 
$$\upsilon_i < \min\limits_{a_{-i} \in A_{-i}}\max\limits_{a_{i} \in A_{i}} u_i(a_i, a_{-i}).$$ Consider a strategy $s_i'$ for \playeri{} that in stage $t \in \mathbb{N}$ plays a best response to the action profile $a_{-i}'$ played by the other players. \playeri{}'s payoff in every stage equals $\max\limits_{a_{i} \in A_{i}} u_i(a_i, a_{-i}')$, and hence this is also the limit-of-means payoff of \playeri{} in $G^\infty$. Because $\max\limits_{a_{i} \in A_{i}} u_i(a_i, a_{-i}') \geq \min\limits_{a_{-i} \in A_{-i}}\max\limits_{a_{i} \in A_{i}} u_i(a_i, a_{-i}) > \upsilon_i$, $s_i$ is not a best response to $s_{-i}$. This contradicts $s$ being a Nash equilibrium.

(2) We will construct a Nash equilibrium $s$ so that, for all $i \in N$, $\lompayoffname(s) = \upsilon_i$. Because $\upsilon$ is feasible, by definition we have, for all $i \in N$, that $\upsilon_i = \sum_{a \in A}{\alpha_a u_i(a)}$ for some $\alpha \in \mathbb{Q}^{|A|}$ with all components non-negative. Each $\alpha_a$ is rational, so we can rewrite it as $\alpha_a=\frac{\beta_a}{\gamma}$ for non-negative integers $\beta_a, \gamma$ satisfying $\sum_{a \in A}{\beta_a} = \gamma.$ 

The payoff vector $\upsilon$ is achieved by splitting $G^\infty$ into periods of $\gamma$ stages. 
Let $s$ be a strategy profile consisting of strategies that, in each period of length $\gamma$, play the action profile $a$ for $\beta_a$ stages for each $a \in A$ in some fixed ordering of $A$, in case no player has deviated. If \playerj{} unilaterally deviates from the prescribed path of play, all other players play $a_{-j}'$ forever, where $a_{-j}'$ is a minmax action against \playerj{}, that is, an action profile satisfying $\max\limits_{a_{j} \in A_{j}} u_j(a_j, a_{-j}') = \min\limits_{a_{-j} \in A_{-j}} \max\limits_{a_{j} \in A_{j}} u_j(a_j, a_{-j})$. 
From the assumption that $\upsilon$ is individually rational we have $\upsilon_j \geq \max\limits_{a_{j} \in A_{j}} u_j(a_j, a_{-j}')$. But $\max\limits_{a_{j} \in A_{j}} u_j(a_j, a_{-j}')$ is an upper bound on the limit-of-means payoff that \playerj{} can obtain if they deviate, and hence any deviation, by any player,
from $s$ cannot yield strictly greater payoff for that player, whence $s$ is a Nash equilibrium. By construction, the limit-of-means payoff of any player \playeri{} on the prescribed path of play is exactly $\upsilon_i$, and hence $s$ leads to the payoff profile $\upsilon.$ Moreover, every strategy in $s$ is computable, as it suffices to iterate over a table of length $\gamma$ to play the next action and compare to the finite history of previous actions played by the other players. 
\end{proof}

Theorem \ref{theorem:folk_avg_spe} was originally proved by Aumann and Shapley \cite{Aumann1994LongTermC}, but without
making computability of the subgame-perfect equilibria explicit in the statement of the result. Below is a proof, using an adaptation of their methods, using the notation and general approach of the present paper:

\begin{proof}[Proof of Theorem \ref{theorem:folk_avg_spe}]
 We modify the proof of Theorem \ref{theorem:folk_avg_ne} such that the strategy profile $s$ is a Nash equilibrium following any finite history. As in that proof, write $\upsilon_i = \sum_{a \in A}{\frac{\beta_a}{\gamma}u_i(a)}$ for each $i \in N$, and define each strategy in $s$ to play the action profile $a$ for $\beta_a$ stages for each $a \in A$ in periods of length $\gamma$ if no player has deviated from $s$.  Define, for each $j \in N$, $\mu_j = \min\limits_{a_{-j} \in A_{-j}} \max\limits_{a_{j} \in A_{j}} u_j(a_j, a_{-j})$.
 Observe that $\mu_j \leq \upsilon_j \leq \max_{a \in A} u_j(a)$.
 If any player \playerj{} deviates at some stage $T$,
 the other players play a minmax action profile against \playerj{}
 (that is, a profile with payoff $\mu_j$ for \playerj{}) for
 $T^2 - T$ stages (we call this a \emph{punishment phase}),
 and then revert to the strategy $s$.

Consider a finite history $h^T \in \mathcal{H}^T_{G^\infty}$ and a subgame $(G^\infty, h^T)$. If all players play $s$ in $(G^\infty, h^T)$, they obtain the payoff profile $\upsilon$ because any deviation in $h^T$ is punished in at most $T^2 - T$ stages following $h^T$. Consider any path of play; if \playerj{} deviates from the prescribed path a finite number of times, let $T_X$ be the stage of their last deviation. After the following $T^2_X - T_X$ stages, the last punishment phase ends and the prescribed path of play continues forever, leading to the payoff profile $\upsilon$, whence \playerj{}'s deviation does not yield strictly greater payoff for \playerj{}. If \playerj{} deviates from the prescribed path infinitely many times, let $T_0, T_1, T_2, \ldots$ be the infinite sequence of stages where \playerj{} deviates. 
 Then, for each $n \in \mathbb{N}$, the payoff in the first $T_n + (T^2_n - T_n) = T^2_n$
stages (that is, at the end of the $n$th punishment phase) is at most:
$$
\frac{1}{T_n + (T^2_n - T_n)}\left(T_n \max_{a \in A}\{u_j(a)\} + (T^2_n - T_n) \mu_j\right)
= \frac{(T^2_n - T_n) \mu_j}{T^2_n} + \frac{T_n \max_{a \in A}\{u_j(a)\}}{T^2_n} \leq \mu_j + O(1/n)
$$
As the payoff of \playerj{} is the limit inferior
of the payoffs after the finite repetitions, the payoff
for \playerj{} when deviating infinitely many times
is at most $\mu_j \leq \upsilon_j$ and hence \playerj{} does not strictly increase their payoff compared to $s$. Hence, as the finite history
$h^T$ was arbitrary, $s$ is a subgame-perfect equilibrium.
%

\end{proof}

\end{document}

%% file: check_triviality.tex
\begin{remark}
Theorem \ref{the:LOM_naive_the} yields a simple criterion for checking whether $G^\infty$ has a computable strategy without a best response: simply check whether the stage game $G$ is non-trivial. 
The payoff matrix of $G$ is a $\vert A_1 \vert \times \vert A_2 \vert$ matrix, and verifying whether $G$ is trivial for \playerone{} amounts to checking the condition $M_1 = \min\limits_{a_{2} \in A_{2}}\max\limits_{a_{1} \in A_{1}} u_1(a_1, a_{2})$. For any $a_2 \in A_2,$ a single scan over $A_1$ gives the value $\max\limits_{a_{1} \in A_{1}} u_1(a_1, a_{2})$. Iterating over all $a_2 \in A_2$ gives $\min\limits_{a_{2} \in A_{2}}\max\limits_{a_{1} \in A_{1}} u_1(a_1, a_{2})$, hence triviality (and thus, non-triviality) can be decided in time $\mathcal{O}(|A_1| \cdot |A_2|)$.
\end{remark}

%% file: convex_hull.tex
Condition \ref{cond:le_last} of Theorem \ref{theorem:final_charac_lom} might at the first glance seem difficult to check, but the Folk theorem provides us with an efficient algorithm for deciding whether there is a Nash equilibrium of $G^\infty$ that is strictly individually rational for \playerone{}. 
By Remark \ref{rem:sheesh_cats}, every payoff profile of $G^\infty$ is in the convex hull $\mathcal{C}$ of the payoff profiles 
of $G$, and by Theorem \ref{theorem:folk_avg_ne}, every Nash equilibrium of $G^\infty$
is individually rational, whence the set of Nash equilibria of $G^\infty$ is a subset of the intersection
of the convex hull $\mathcal{C}$ and the set 
$$
\mathcal{Q} = \{(x,y) \in \mathbb{R}^2 : x \geq \min\limits_{a_{2} \in A_{2}}\max\limits_{a_{1} \in A_{1}} u_1(a_1, a_2), y \geq \min\limits_{a_{1} \in A_{1}}\max\limits_{a_{2} \in A_{2}} u_2(a_1, a_2)\}
$$
But also by Theorem \ref{theorem:folk_avg_ne}, for every feasible and individually rational payoff profile $\upsilon$, there is a Nash equilibrium $s$ of $G^\infty$ with payoff profile $\upsilon$. Hence,
every payoff profile $\upsilon \in \mathbb{Q}^2$ in $\mathcal{C} \cap \mathcal{Q}$ is a Nash equilibrium.

To verify condition \ref{cond:le_last}, it thus suffices to consider the various cases of $\mathcal{C} \cap \mathcal{Q}$.
If  $\mathcal{C} \cap \mathcal{Q} = \{(p,q)\}$, it is one of the corner points of $\mathcal{C}$ or $\mathcal{Q}$, hence either the minmax profile or one of the payoff profiles of $G$, and hence has rational components; thus, $(p,q)$ is a Nash equilibrium,
and we can check directly if $p > \min\limits_{a_{2} \in A_{2}}\max\limits_{a_{1} \in A_{1}} u_1(a_1, a_2)$.
If $\mathcal{C} \cap \mathcal{Q}$ is a line segment $L$, the fact that $\mathcal{Q}$ is an upper-right quarter-plane (hence have edges parallel to the $x$- and $y$-axes) entail that $L$ is either a subset of $Z_1 = \{(x,y) : x \geq \min\limits_{a_{2} \in A_{2}}\max\limits_{a_{1} \in A_{1}} u_1(a_1, a_2)\}$, or of
$Z_2 = \{(x,y) : y \geq \min\limits_{a_{1} \in A_{1}}\max\limits_{a_{2} \in A_{2}} u_2(a_1, a_2)\}$; if $L$ is a subset of $Z_1$ (clearly checkable by testing a single point), there are no Nash equilibria
$s$ with $\lompayoffname_1(s) > \min\limits_{a_{2} \in A_{2}}\max\limits_{a_{1} \in A_{1}} u_1(a_1, a_2)$,
and if $L$ is a subset of $Z_2$, density
of $\mathbb{Q}^2$ in $\mathbb{R}^2$, entails that $L$ contains a point with rational components, which then by Theorem \ref{theorem:folk_avg_ne} corresponds to a Nash equilibrium $s$ with 
$\lompayoffname_1(s) >  \min\limits_{a_{2} \in A_{2}}\max\limits_{a_{1} \in A_{1}} u_1(a_1, a_2)$,
whence \ref{cond:le_last} holds. Finally, if $\mathcal{C} \cap \mathcal{Q}$ is neither a singleton, nor a line segment, it is itself a convex polygon, and by density of $\mathbb{Q}^2$ in $\mathbb{R}^2$ is contains
a point $(p,q) \in \mathbb{Q}^2$ with $p > \min\limits_{a_{2} \in A_{2}}\max\limits_{a_{1} \in A_{1}} u_1(a_1, a_2)$ that is a Nash equilibrium by Theorem \ref{theorem:folk_avg_ne}.

Computing $\mathcal{C} \cap \mathcal{Q}$ can be performed by first restricting $\mathcal{Q}$ to a sufficiently large rectangle, for example the rectangle with lower-left corner
$$
\left(\min\limits_{a_{2} \in A_{2}}\max\limits_{a_{1} \in A_{1}} u_1(a_1, a_2),\min\limits_{a_{1} \in A_{1}}\max\limits_{a_{2} \in A_{2}} u_2(a_1, a_2)\right)
$$ 
and upper-right corner
\begin{align*}
\left(\max\limits_{a \in A} u_1(a), \max\limits_{a \in A} u_2(a)\right)
\end{align*}
and subsequently using a standard algorithm for computing the intersection of convex polygons.

Using, e.g., Chan's algorithm \cite{Chan1996} for finding a convex hull of a set of $n$ points runs in  $\mathcal{O}(n\log{h})$ time, where $h \leq n$ denotes the number of points in the convex hull,
and by any number of classical algorithms, e.g. \cite{Shamos1975}, the intersection of two convex polygons of size at most $n$ can be computed in $\mathcal{O}(n)$ time. 

Hence, for a payoff matrix of dimension $n \times m$, deciding whether condition \ref{cond:le_last} holds can be done in $\mathcal{O}(nm \log{nm})$ time using the method described above.

%% file: examples.tex
To illustrate our results, we give examples of well-known games that satisfy different criteria in Theorems \ref{the:LOM_naive_the} and \ref{theorem:final_charac_lom}.

\begin{example}[Rock-paper-scissors]
\label{ex:rockpaperscissors}
Rock-paper-scissors is a two-player game with $A_1 = A_2 = \{ \textnormal{Rock, Paper, Scissors}\}$ and payoff matrix as follows:
\begin{center}
\begin{tabular}{cccc}
                              & Rock                       & Paper                      & Scissors                   \\ \cline{2-4} 
\multicolumn{1}{c|}{Rock}     & \multicolumn{1}{c|}{0, 0}  & \multicolumn{1}{c|}{-1, 1} & \multicolumn{1}{c|}{1, -1} \\ \cline{2-4} 
\multicolumn{1}{c|}{Paper}    & \multicolumn{1}{c|}{1, -1} & \multicolumn{1}{c|}{0, 0}  & \multicolumn{1}{c|}{-1, 1} \\ \cline{2-4} 
\multicolumn{1}{c|}{Scissors} & \multicolumn{1}{c|}{-1, 1} & \multicolumn{1}{c|}{1, -1} & \multicolumn{1}{c|}{0, 0}  \\ \cline{2-4} 
\end{tabular}
\end{center}

The minmax payoff profile of Rock-paper-scissors is $(1,1),$ while $1$ is also the maximum payoff that \playerone{} can obtain. Hence, Rock-paper-scissors is trivial for \playerone{} and by Theorem
\ref{the:LOM_naive_the}, every computable strategy of \playertwo{} has a computable best response.
\end{example}

\begin{example}[Deadlock]
\label{ex:deadlock} Deadlock is a two-player game $G$ with $A_1 = A_2 = \{ C, D\}$ and payoff matrix as follows:
\begin{center}
\begin{tabular}{ccc}
\textit{}                       & \textit{C}                         & \textit{D}                         \\ \cline{2-3} 
\multicolumn{1}{c|}{\textit{C}} & \multicolumn{1}{c|}{\textit{1, 1}} & \multicolumn{1}{c|}{\textit{0, 3}} \\ \cline{2-3} 
\multicolumn{1}{c|}{\textit{D}} & \multicolumn{1}{c|}{\textit{3, 0}} & \multicolumn{1}{c|}{\textit{2, 2}} \\ \cline{2-3} 
\end{tabular}
\end{center}

The minmax payoff profile of Deadlock is $(2,2)$ but \playerone{} cannot obtain higher payoff than $2$ under a Nash equilibrium. By Theorem \ref{the:LOM_naive_the}, there is a computable strategy of \playertwo{} without a computable best response, but by Theorem \ref{theorem:final_charac_lom}, no such strategy is part of any strategy profile that is a Nash equilibrium.
\end{example}

\begin{example}[Stag hunt]
\label{ex:staghunt} Stag hunt is a two-player game $G$ with $A_1 = A_2 = \{ \textnormal{Stag, Hare}\}$ and the following payoff matrix:
\begin{center}
\begin{tabular}{ccc}
\textit{}                       & \textit{Stag}                         & \textit{Hare}                         \\ \cline{2-3} 
\multicolumn{1}{c|}{\textit{Stag}} & \multicolumn{1}{c|}{\textit{3, 3}} & \multicolumn{1}{c|}{\textit{0, 2}} \\ \cline{2-3} 
\multicolumn{1}{c|}{\textit{Hare}} & \multicolumn{1}{c|}{\textit{2, 0}} & \multicolumn{1}{c|}{\textit{1, 1}} \\ \cline{2-3} 
\end{tabular}
\end{center}

The minmax payoff profile of Stag hunt is $(1,1)$, and the repeated play of $(\textnormal{Stag, Stag})$ is a Nash equilibrium of $G^\infty$ with \playerone{}'s payoff being $3.$ By Theorem \ref{theorem:final_charac_lom}, there is a subgame-perfect equilibrium $s=(s_1,s_2)$
(hence also a Nash equilibrium) of $G^\infty$ such that $s_2$ is a computable strategy that does not have computable best response. 
\end{example}

%% file: final.bbl

\begin{thebibliography}{27}


\ifx \showCODEN    \undefined \def \showCODEN     #1{\unskip}     \fi
\ifx \showDOI      \undefined \def \showDOI       #1{#1}\fi
\ifx \showISBNx    \undefined \def \showISBNx     #1{\unskip}     \fi
\ifx \showISBNxiii \undefined \def \showISBNxiii  #1{\unskip}     \fi
\ifx \showISSN     \undefined \def \showISSN      #1{\unskip}     \fi
\ifx \showLCCN     \undefined \def \showLCCN      #1{\unskip}     \fi
\ifx \shownote     \undefined \def \shownote      #1{#1}          \fi
\ifx \showarticletitle \undefined \def \showarticletitle #1{#1}   \fi
\ifx \showURL      \undefined \def \showURL       {\relax}        \fi
\providecommand\bibfield[2]{#2}
\providecommand\bibinfo[2]{#2}
\providecommand\natexlab[1]{#1}
\providecommand\showeprint[2][]{arXiv:#2}

\bibitem[\protect\citeauthoryear{Arrow and Debreu}{Arrow and Debreu}{1954}]%
        {ArrowDebreu}
\bibfield{author}{\bibinfo{person}{Kenneth Arrow} {and} \bibinfo{person}{Gerard
  Debreu}.} \bibinfo{year}{1954}\natexlab{}.
\newblock \showarticletitle{Existence of an Equilibrium for a Competitive
  Economy}.
\newblock \bibinfo{journal}{\emph{Econometrica}}  \bibinfo{volume}{22}
  (\bibinfo{year}{1954}), \bibinfo{pages}{265--290}.
\newblock


\bibitem[\protect\citeauthoryear{Aumann}{Aumann}{1981}]%
        {bib:aumann1981}
\bibfield{author}{\bibinfo{person}{Robert~J. Aumann}.}
  \bibinfo{year}{1981}\natexlab{}.
\newblock \showarticletitle{Survey of repeated games}.
\newblock \bibinfo{journal}{\emph{Essays in game theory and mathematical
  economics in honor of Oskar Morgenstern}} (\bibinfo{year}{1981}),
  \bibinfo{pages}{11--42}.
\newblock


\bibitem[\protect\citeauthoryear{Aumann and Shapley}{Aumann and
  Shapley}{1994}]%
        {Aumann1994LongTermC}
\bibfield{author}{\bibinfo{person}{Robert~J. Aumann} {and}
  \bibinfo{person}{Lloyd~S. Shapley}.} \bibinfo{year}{1994}\natexlab{}.
\newblock \showarticletitle{Long-Term Competition - A Game-Theoretic Analysis}.
  In \bibinfo{booktitle}{\emph{Essays in Game Theory}}.
\newblock


\bibitem[\protect\citeauthoryear{Ben-{P}orath}{Ben-{P}orath}{1990}]%
        {BENPORATH19901}
\bibfield{author}{\bibinfo{person}{Elchanan Ben-{P}orath}.}
  \bibinfo{year}{1990}\natexlab{}.
\newblock \showarticletitle{The complexity of computing a best response
  automaton in repeated games with mixed strategies}.
\newblock \bibinfo{journal}{\emph{Games and Economic Behavior}}
  \bibinfo{volume}{2}, \bibinfo{number}{1} (\bibinfo{year}{1990}),
  \bibinfo{pages}{1 -- 12}.
\newblock


\bibitem[\protect\citeauthoryear{Berg and Kitti}{Berg and Kitti}{2019}]%
        {Berg2011}
\bibfield{author}{\bibinfo{person}{Kimmo Berg} {and} \bibinfo{person}{Mitri
  Kitti}.} \bibinfo{year}{2019}\natexlab{}.
\newblock \showarticletitle{Equilibrium paths in discounted supergames}.
\newblock \bibinfo{journal}{\emph{Discrete Applied Mathematics}}
  \bibinfo{volume}{260} (\bibinfo{year}{2019}), \bibinfo{pages}{1 -- 27}.
\newblock


\bibitem[\protect\citeauthoryear{Borgs, Chayes, Immorlica, Kalai, Mirrokni, and
  Papadimitriou}{Borgs et~al\mbox{.}}{2010}]%
        {BORGS201034}
\bibfield{author}{\bibinfo{person}{Christian Borgs}, \bibinfo{person}{Jennifer
  Chayes}, \bibinfo{person}{Nicole Immorlica}, \bibinfo{person}{Adam~Tauman
  Kalai}, \bibinfo{person}{Vahab Mirrokni}, {and} \bibinfo{person}{Christos
  Papadimitriou}.} \bibinfo{year}{2010}\natexlab{}.
\newblock \showarticletitle{The myth of the Folk Theorem}.
\newblock \bibinfo{journal}{\emph{Games and Economic Behavior}}
  \bibinfo{volume}{70}, \bibinfo{number}{1} (\bibinfo{year}{2010}),
  \bibinfo{pages}{34 -- 43}.
\newblock


\bibitem[\protect\citeauthoryear{Chan}{Chan}{1996}]%
        {Chan1996}
\bibfield{author}{\bibinfo{person}{T.~M. Chan}.}
  \bibinfo{year}{1996}\natexlab{}.
\newblock \showarticletitle{Optimal Output-Sensitive Convex Hull Algorithms in
  Two and Three Dimensions}.
\newblock \bibinfo{journal}{\emph{Discrete Comput. Geom.}}
  \bibinfo{volume}{16}, \bibinfo{number}{4} (\bibinfo{year}{1996}),
  \bibinfo{pages}{361–368}.
\newblock


\bibitem[\protect\citeauthoryear{Chen, Lin, Tang, Wang, Wang, and Wang}{Chen
  et~al\mbox{.}}{2017}]%
        {Chen-kmemory2017}
\bibfield{author}{\bibinfo{person}{Lijie Chen}, \bibinfo{person}{Fangzhen Lin},
  \bibinfo{person}{Pingzhong Tang}, \bibinfo{person}{Kangning Wang},
  \bibinfo{person}{Ruosong Wang}, {and} \bibinfo{person}{Shiheng Wang}.}
  \bibinfo{year}{2017}\natexlab{}.
\newblock \showarticletitle{K-Memory Strategies in Repeated Games}. In
  \bibinfo{booktitle}{\emph{Proceedings of the 16th Conference on Autonomous
  Agents and MultiAgent Systems}} \emph{(\bibinfo{series}{AAMAS '17})}.
  \bibinfo{publisher}{International Foundation for Autonomous Agents and
  Multiagent Systems}, \bibinfo{address}{Richland, SC},
  \bibinfo{pages}{1493--1498}.
\newblock


\bibitem[\protect\citeauthoryear{Chen and Tang}{Chen and Tang}{2015}]%
        {bib:Chen:2015}
\bibfield{author}{\bibinfo{person}{Lijie Chen} {and} \bibinfo{person}{Pingzhong
  Tang}.} \bibinfo{year}{2015}\natexlab{}.
\newblock \showarticletitle{Bounded Rationality of Restricted {T}uring
  Machines}. In \bibinfo{booktitle}{\emph{Proceedings of the 2015 International
  Conference on Autonomous Agents and Multiagent Systems}}
  \emph{(\bibinfo{series}{AAMAS '15})}. \bibinfo{publisher}{International
  Foundation for Autonomous Agents and Multiagent Systems},
  \bibinfo{address}{Richland, SC}, \bibinfo{pages}{1673--1674}.
\newblock


\bibitem[\protect\citeauthoryear{Fortnow and Whang}{Fortnow and Whang}{1994}]%
        {DBLP:conf/stoc/FortnowW94}
\bibfield{author}{\bibinfo{person}{Lance Fortnow} {and} \bibinfo{person}{Duke
  Whang}.} \bibinfo{year}{1994}\natexlab{}.
\newblock \showarticletitle{Optimality and domination in repeated games with
  bounded players}. In \bibinfo{booktitle}{\emph{Proceedings of the
  Twenty-Sixth Annual {ACM} Symposium on Theory of Computing (STOC 1994)}}.
  \bibinfo{pages}{741--749}.
\newblock


\bibitem[\protect\citeauthoryear{Fudenberg and Tirole}{Fudenberg and
  Tirole}{1991}]%
        {bib:ft1991}
\bibfield{author}{\bibinfo{person}{Drew Fudenberg} {and} \bibinfo{person}{Jean
  Tirole}.} \bibinfo{year}{1991}\natexlab{}.
\newblock \bibinfo{booktitle}{\emph{Game Theory}}.
\newblock \bibinfo{publisher}{MIT Press}, \bibinfo{address}{Cambridge, MA}.
\newblock


\bibitem[\protect\citeauthoryear{Gilboa}{Gilboa}{1988}]%
        {GILBOA1988342}
\bibfield{author}{\bibinfo{person}{Itzhak Gilboa}.}
  \bibinfo{year}{1988}\natexlab{}.
\newblock \showarticletitle{The complexity of computing best-response automata
  in repeated games}.
\newblock \bibinfo{journal}{\emph{Journal of Economic Theory}}
  \bibinfo{volume}{45}, \bibinfo{number}{2} (\bibinfo{year}{1988}),
  \bibinfo{pages}{342 -- 352}.
\newblock


\bibitem[\protect\citeauthoryear{Halpern, Pass, and Seeman}{Halpern
  et~al\mbox{.}}{2019}]%
        {DBLP:journals/geb/HalpernPS19}
\bibfield{author}{\bibinfo{person}{Joseph~Y. Halpern}, \bibinfo{person}{Rafael
  Pass}, {and} \bibinfo{person}{Lior Seeman}.} \bibinfo{year}{2019}\natexlab{}.
\newblock \showarticletitle{The truth behind the myth of the Folk theorem}.
\newblock \bibinfo{journal}{\emph{Games and Economic Behavior}}
  \bibinfo{volume}{117} (\bibinfo{year}{2019}), \bibinfo{pages}{479--498}.
\newblock


\bibitem[\protect\citeauthoryear{Jones}{Jones}{1997}]%
        {Jones1997}
\bibfield{author}{\bibinfo{person}{Neil~D. Jones}.}
  \bibinfo{year}{1997}\natexlab{}.
\newblock \bibinfo{booktitle}{\emph{Computability and Complexity: From a
  Programming Perspective}}.
\newblock \bibinfo{publisher}{MIT Press}, \bibinfo{address}{Cambridge, MA,
  USA}.
\newblock
\showISBNx{0-262-10064-9}


\bibitem[\protect\citeauthoryear{Knoblauch}{Knoblauch}{1994}]%
        {bib:knoblauch1994}
\bibfield{author}{\bibinfo{person}{Vicki Knoblauch}.}
  \bibinfo{year}{1994}\natexlab{}.
\newblock \showarticletitle{Computable Strategies for Repeated Prisoner's
  Dilemma}.
\newblock \bibinfo{journal}{\emph{Games and Economic Behavior}}
  \bibinfo{volume}{7} (\bibinfo{year}{1994}), \bibinfo{pages}{381--389}.
\newblock


\bibitem[\protect\citeauthoryear{Kozen}{Kozen}{1980}]%
        {Kozen1980}
\bibfield{author}{\bibinfo{person}{Dexter Kozen}.}
  \bibinfo{year}{1980}\natexlab{}.
\newblock \showarticletitle{Indexings of subrecursive classes}.
\newblock \bibinfo{journal}{\emph{Theoretical Computer Science}}
  \bibinfo{volume}{11}, \bibinfo{number}{3} (\bibinfo{year}{1980}),
  \bibinfo{pages}{277 -- 301}.
\newblock


\bibitem[\protect\citeauthoryear{Leyton-Brown and Shoham}{Leyton-Brown and
  Shoham}{2008}]%
        {bib:Leyton-Brown:2008:EGT:1481632}
\bibfield{author}{\bibinfo{person}{Kevin Leyton-Brown} {and}
  \bibinfo{person}{Yoav Shoham}.} \bibinfo{year}{2008}\natexlab{}.
\newblock \bibinfo{booktitle}{\emph{Essentials of Game Theory: A Concise,
  Multidisciplinary Introduction} (\bibinfo{edition}{1st} ed.)}.
\newblock \bibinfo{publisher}{Morgan and Claypool Publishers}.
\newblock


\bibitem[\protect\citeauthoryear{Nachbar and Zame}{Nachbar and Zame}{1996}]%
        {bib:Nachbar1996}
\bibfield{author}{\bibinfo{person}{John~H. Nachbar} {and}
  \bibinfo{person}{William~R. Zame}.} \bibinfo{year}{1996}\natexlab{}.
\newblock \showarticletitle{Non-computable strategies and discounted repeated
  games}.
\newblock \bibinfo{journal}{\emph{Economic Theory}} \bibinfo{volume}{8},
  \bibinfo{number}{1} (\bibinfo{year}{1996}), \bibinfo{pages}{103--122}.
\newblock


\bibitem[\protect\citeauthoryear{Neyman and Okada}{Neyman and Okada}{2000}]%
        {DBLP:journals/ijgt/NeymanO00}
\bibfield{author}{\bibinfo{person}{Abraham Neyman} {and}
  \bibinfo{person}{Daijiro Okada}.} \bibinfo{year}{2000}\natexlab{}.
\newblock \showarticletitle{Two-person repeated games with finite automata}.
\newblock \bibinfo{journal}{\emph{Int. J. Game Theory}} \bibinfo{volume}{29},
  \bibinfo{number}{3} (\bibinfo{year}{2000}), \bibinfo{pages}{309--325}.
\newblock


\bibitem[\protect\citeauthoryear{Osborne and Rubinstein}{Osborne and
  Rubinstein}{1994}]%
        {Osborne1994}
\bibfield{author}{\bibinfo{person}{Martin~J. Osborne} {and}
  \bibinfo{person}{Ariel Rubinstein}.} \bibinfo{year}{1994}\natexlab{}.
\newblock \bibinfo{booktitle}{\emph{A course in game theory}}.
\newblock \bibinfo{publisher}{The MIT Press}, \bibinfo{address}{Cambridge,
  USA}.
\newblock
\newblock
\shownote{electronic edition.}


\bibitem[\protect\citeauthoryear{Richter and Wong}{Richter and Wong}{1999}]%
        {RePEc:spr:joecth:v:14:y:1999:i:1:p:1-27}
\bibfield{author}{\bibinfo{person}{Marcel~K. Richter} {and}
  \bibinfo{person}{Kam-Chau Wong}.} \bibinfo{year}{1999}\natexlab{}.
\newblock \showarticletitle{{Non-computability of competitive equilibrium}}.
\newblock \bibinfo{journal}{\emph{Economic Theory}} \bibinfo{volume}{14},
  \bibinfo{number}{1} (\bibinfo{year}{1999}), \bibinfo{pages}{1--27}.
\newblock


\bibitem[\protect\citeauthoryear{Rogers}{Rogers}{1967}]%
        {Rogers}
\bibfield{author}{\bibinfo{person}{Hartley Rogers}.}
  \bibinfo{year}{1967}\natexlab{}.
\newblock \bibinfo{booktitle}{\emph{Theory of {R}ecursive {F}unctions and
  {E}ffective {C}omputability}}.
\newblock \bibinfo{publisher}{McGraw-Hill}.
\newblock
\newblock
\shownote{Reprint, MIT press 1987.}


\bibitem[\protect\citeauthoryear{Rubinstein}{Rubinstein}{1986}]%
        {RUBINSTEIN198683}
\bibfield{author}{\bibinfo{person}{Ariel Rubinstein}.}
  \bibinfo{year}{1986}\natexlab{}.
\newblock \showarticletitle{Finite automata play the repeated prisoner's
  dilemma}.
\newblock \bibinfo{journal}{\emph{Journal of Economic Theory}}
  \bibinfo{volume}{39}, \bibinfo{number}{1} (\bibinfo{year}{1986}),
  \bibinfo{pages}{83 -- 96}.
\newblock


\bibitem[\protect\citeauthoryear{Shamos}{Shamos}{1975}]%
        {Shamos1975}
\bibfield{author}{\bibinfo{person}{Michael~Ian Shamos}.}
  \bibinfo{year}{1975}\natexlab{}.
\newblock \showarticletitle{Geometric Complexity}. In
  \bibinfo{booktitle}{\emph{Proceedings of the Seventh Annual ACM Symposium on
  Theory of Computing}} \emph{(\bibinfo{series}{STOC ’75})}.
  \bibinfo{pages}{224–233}.
\newblock


\bibitem[\protect\citeauthoryear{Sipser}{Sipser}{2013}]%
        {bib:sipser}
\bibfield{author}{\bibinfo{person}{Michael Sipser}.}
  \bibinfo{year}{2013}\natexlab{}.
\newblock \bibinfo{booktitle}{\emph{Introduction to the Theory of Computation}
  (\bibinfo{edition}{3rd international} ed.)}.
\newblock \bibinfo{publisher}{Cengage Learning}.
\newblock


\bibitem[\protect\citeauthoryear{Smullyan}{Smullyan}{1958}]%
        {bib:Smullyan1958}
\bibfield{author}{\bibinfo{person}{Raymond~M. Smullyan}.}
  \bibinfo{year}{1958}\natexlab{}.
\newblock \showarticletitle{Undecidability and recursive inseparability}.
\newblock \bibinfo{journal}{\emph{Mathematical Logic Quarterly}}
  \bibinfo{volume}{4}, \bibinfo{number}{7‐11} (\bibinfo{year}{1958}),
  \bibinfo{pages}{143--147}.
\newblock


\bibitem[\protect\citeauthoryear{Zuo and Tang}{Zuo and Tang}{2015}]%
        {DBLP:conf/aaai/ZuoT15}
\bibfield{author}{\bibinfo{person}{Song Zuo} {and} \bibinfo{person}{Pingzhong
  Tang}.} \bibinfo{year}{2015}\natexlab{}.
\newblock \showarticletitle{Optimal Machine Strategies to Commit to in
  Two-Person Repeated Games}. In \bibinfo{booktitle}{\emph{Proceedings of the
  Twenty-Ninth {AAAI} Conference on Artificial Intelligence, January 25-30,
  2015, Austin, Texas, {USA}}}. \bibinfo{pages}{1071--1078}.
\newblock


\end{thebibliography}
